\documentclass{article}

\usepackage[utf8]{inputenc}
\usepackage[a4paper,hmargin=2.75cm,vmargin=2.5cm,bindingoffset=0.5cm]{geometry}
\usepackage[ruled,vlined]{algorithm2e}
\usepackage{tikz}
\usepackage{setspace}
\usepackage{array}
\usepackage{ragged2e}
\usepackage{amssymb}

\usepackage[title]{appendix}

\usepackage{amsthm}

\usepackage{mathtools}
\usepackage{enumerate}
\usepackage{caption}

\usetikzlibrary{decorations.pathreplacing} 

\usepackage[colorlinks,citecolor=blue]{hyperref}

\usepackage{doi} 

\usetikzlibrary{shapes,positioning,fit,matrix,arrows,arrows.meta}
\usetikzlibrary{decorations.pathmorphing,decorations.pathreplacing}

\newtheorem{theorem}{Theorem}
\newtheorem{lemma}{Lemma}
\newtheorem{observation}{Observation}
\newtheorem{corollary}{Corollary}

\newtheorem{innercustomlem}{Lemma}
\newenvironment{customlem}[1]
{\renewcommand\theinnercustomlem{#1 (Restated)}\innercustomlem}
{\endinnercustomlem}

\newtheorem{innercustomthm}{Theorem}
\newenvironment{customthm}[1]
{\renewcommand\theinnercustomthm{#1 (Restated)}\innercustomthm}
{\endinnercustomthm}

\newtheoremstyle{claim}{5pt}{5pt}{}{}{\bf}{:}{.5em}{\thmname{#1}\thmnumber{ #2}\thmnote{ (#3)}}

\theoremstyle{claim}
\newtheorem{claim-alt}{Claim}

\newtheoremstyle{freethm}{3pt}{3pt}{}{}{\bfseries}{}{.5em}{\thmname{#1}\thmnumber{ #2}\thmnote{ (#3)}.\\}

\theoremstyle{freethm}
\newtheorem{construct}{Construction}
\newtheorem{remark}{Remark}

\theoremstyle{claim}
\newtheorem{innercustomclaim}{Claim}
\newenvironment{customclaim}[1]
{\renewcommand\theinnercustomclaim{#1 (Restated)}\innercustomclaim}
{\endinnercustomclaim}

\newcommand{\openpack}{\textsc{Open Packing}}

\newcommand{\tdset}{\textsc{Total Dominating Set}}

\newcommand{\maxopenpack}{\textsc{Max-Open Packing}}

\newcommand{\mintdset}{\textsc{Min-Total Dominating Set}}

\newcommand{\indset}{\textsc{Independent Set}}
\newcommand{\maxindset}{\textsc{Max-Independent Set}}

\tikzset{
	block1/.style={circle,inner sep=1.25pt, minimum size=0.75ex, draw =black,fill=white},
	box1/.style={rectangle, minimum size=0.75ex,align=center, draw =black,fill=white}
}

\begin{document}
\title{Open Packing in Graphs: Bounds and Complexity}
%
%
\author{Shalu M.\ A. 
	\and Kirubakaran V.\ K. 
	}
	\date{\small Indian Institute of Information Technology, Design and Manufacturing,\\ \small Kancheepuram, Chennai, India.}
%
%
\maketitle              
\begin{abstract}
Given a graph $G(V,E)$, a vertex subset $S$ of $G$ is called an open packing in $G$ if no pair of distinct vertices in $S$ have a common neighbour in $G$. The size of a largest open packing in $G$ is called the open packing number of $G$ and is denoted by $\rho^o(G)$. It would be interesting to note that the open packing number is a lower bound for the total domination number in graphs with no isolated vertices [Henning and Slater, 1999]. Given a graph $G$ and a positive integer $k$, the decision problem \openpack\ tests whether $G$ has an open packing of size at least $k$. The optimization problem \maxopenpack\ takes a graph $G$ as input and finds the open packing number of $G$.\\
It is known that \openpack\ is NP-complete on split graphs (i.e., the class of $\{2K_2,C_4,C_5\}$-free graphs) [Ramos et al., 2014]. In this work, we complete the study on the complexity (P vs NPC) of \openpack\ on $H$-free graphs for every graph $H$ with at least three vertices by proving that \openpack\ is (i) NP-complete on $K_{1,3}$-free graphs and (ii) polynomial time solvable on $(P_4\cup rK_1)$-free graphs for every $r\geq 1$. In the course of proving (ii), we show that for every $t\in \{2,3,4\}$ and $r\geq 1$, if $G$ is a $(P_t\cup rK_1)$-free graph, then $\rho^o(G)$ is bounded above by a linear function of $r$. In addition, we find near-optimal upper bounds for the total domination number in the class of $(P_t\cup rK_1)$-free graphs for every $t\in \{3,4\}$ and $r\geq 1$. Moreover, we show that \openpack\ parameterized by solution size is W[1]-complete on $K_{1,3}$-free graphs, and \maxopenpack\ is hard to approximate within a factor of $\displaystyle (n^{\frac{1}{2}-\delta})$ for any $\delta>0$ on $K_{1,3}$-free graphs unless P = NP. Further, we prove that \openpack\ is (a) NP-complete on $K_{1,4}$-free split graphs and (b) polynomial time solvable on $K_{1,3}$-free split graphs. We prove a similar dichotomy result on split graphs with degree restrictions on the vertices in the independent set of the clique-independent set partition of the split graphs.\\
 
\end{abstract}
\section{Introduction}
In a graph $G(V,E)$, a vertex subset $D$ of $G$ is called a {\it total dominating set} in $G$ if every vertex in $V(G)$ is adjacent to some vertex in $D$. In other words, $V(G)=\cup_{u\in D}N_G(u)$, where $N_G(u)$ denotes the set of vertices in $V(G)$ that are adjacent to $u$ in $G$. This implies that if $D$ is a total dominating set in $G$, then $|D\cap N_G(u)|\geq 1$ for every vertex $u\in V(G)$. Note that by the definition of a total dominating set, it is evident that a graph $G$ admits a total dominating set if and only if $G$ has no isolated vertices. The cardinality of a smallest total dominating set in $G$ is called the {\it total domination number}, $\gamma_t(G)$, of $G$. The total domination problem is useful in facility location problems, monitoring computer networks, and electronic communications~\cite{fundamentalsofdomination1998}. In this article, we study the total dominating set problem and its dual problem, called open packing in graphs \cite{Henning1998}. A vertex subset $S$ of a graph $G$ is called an {\it open packing} in $G$ if, for every pair of distinct vertices $x,y\in S$, $N_G(x)\cap N_G(y)$ is empty. The {\it open packing number} of $G$, denoted by $\rho^o(G)$, is the size of a largest open packing in $G$. Note that if $S$ is an open packing in $G$, then $|S\cap N_G(u)|\leq 1$ for every $u\in V(G)$. This validates the primal-dual relation between the total dominating set and the open packing, and thus $\rho^o(G)\leq \gamma_t(G)$. The decision and the optimisation versions of the problems of total dominating set and open packing are as follows. 
\begin{center}
	\begin{tabular}{|m{8.75cm}| }
		\hline \\[-9pt]
		\tdset\ \\[2pt]
		\hline \\[-9pt]
		\textbf {Instance\hspace{.08cm}: }A graph $G$ and a positive integer $k\leq |V(G)|$.\\
		\textbf {Question: }Does $G$ has a total dominating set of size\\ \hspace{1.93cm}  at most $k$? \\
		\hline
	\end{tabular}
	\hspace{1.5mm}
	\begin{tabular}{|m{4.9cm}| }
		\hline \\[-9pt]
		\mintdset\ \\[2pt]
		\hline \\[-9pt]
		\textbf {Instance: }A graph $G$.\\
		\textbf {Goal \hspace{.5cm}: }Find $\gamma_t(G)$. \\
		~\\
		\hline
	\end{tabular}
	\vspace{3mm}
	\\
	\begin{tabular}{|m{8.75cm}| }
		\hline \\[-9pt]
		\openpack\ \\[2pt]
		\hline \\[-9pt]
		\textbf {Instance\hspace{.08cm}: }A graph $G$ and a positive integer $k\leq |V(G)|$.\\
		\textbf {Question: }Is there an open packing of size at least $k$ \\ \hspace{1.8cm} in $G$? \\
		\hline
	\end{tabular}
	\hspace{1.5mm}
	\begin{tabular}{|m{4.9cm}| }
		\hline \\[-9pt]
		\maxopenpack\ \\[2pt]
		\hline \\[-9pt]
		\textbf {Instance: }A graph $G$.\\
		\textbf {Goal \hspace{.5cm}: }Find $\rho^o(G)$. \\
		~\\
		\hline
	\end{tabular}
\end{center} 
Total dominating set is one of the well-studied problems in the literature and an extensive list of results can be found in \cite{parameterized-downey-2013,fundamentalsofdomination1998,hennings2013-td}. In the context of our research, it would be interesting to note that \tdset\ is NP-complete on $K_{1,5}$-free split graphs~\cite{White1995} and polynomial time solvable on $K_{1,4}$-free split graphs~\cite{RENJITH2020246}. Also, \tdset\ is NP-complete on $K_{1,3}$-free graphs~\cite{McRae1995} and an optimal total dominating set of a chordal bipartite graph can be found in polynomial time~\cite{damaschke}. The study on open packing of graphs was initiated by Henning and Slater~\cite{Henning1999}, and Rall~\cite{Rall2005} proved that for every non-trivial tree $T$, $\gamma_t(G)=\rho^o(G)$. In a recent work~\cite{publish} (yet to be published), we extended this result by proving that the total domination number and the open packing number are equal when the underlying graph is a chordal bipartite graph with no isolated vertices. It is also known that \openpack\ is NP-complete (NPC in short) on split graphs (equivalently, the class of $\{2K_2,C_4,C_5\}$-free graphs)~\cite{Igor2014} and planar bipartite graphs of maximum degree three (a subclass of $K_3$-free graphs)~\cite{publish}. In this work, we complete the study of \openpack\ on $H$-free graphs for every graph $H$ with at least three vertices by proving the following theorems.
\begin{theorem}
	Let $H$ be a graph on three vertices. Then, an optimal open packing in $H$-free graphs can be found in polynomial time if and only if $H\ncong K_3$ unless P = NP.
	\label{h-on-three-vertices}
\end{theorem}
\begin{theorem}
	For $p\geq 4$, let $H$ be a graph on $p$ vertices. Then, \openpack\ is polynomial time solvable on the class of $H$-free graphs if and only if $H\in \{pK_1,(K_2\cup (p-2)K_1),(P_3\cup (p-3)K_1),(P_4\cup(p-4)K_1)\}$ unless P = NP.
	\label{thm-h-dic}
\end{theorem}
\noindent To prove the above theorem, we proved the following results.
\begin{enumerate}
	\item[(i)] \openpack\ is NP-complete on $K_{1,3}$-free graphs.
	\item[(ii)] For $t\in \{2,3,4\}$ and $r\geq 1$, if $G$ is a connected $(P_t\cup rK_1)$-free graph, then
	\begin{itemize}
		\item[(a)] $\rho^o(G)\leq 2r+1$ and $\gamma_t(G)\leq 2r+2$ provided $t=4$.
		\item[(b)] $\rho^o(G)\leq 2r$ and $\gamma_t(G)\leq 2r+1$ provided $t=3$.
		\item[(c)] $\rho^o(G)\leq \max\{r+1,2(r-1)\}$ provided $t=2$.
	\end{itemize} 
\end{enumerate}
We also show that the bounds on $\rho^o(G)$ given in (ii) are tight. Moreover, we use the bounds on the total domination number in (a) and (b) to show that the results equivalent to those of Theorems~\ref{h-on-three-vertices} and \ref{thm-h-dic} hold for \tdset\ in $H$-free graphs as well (see Theorem~\ref{thm-h-dic-tdset}). This implies that \tdset\ and \openpack\ are of the same nature in the view of classical complexity (P vs NPC) in $H$-free graphs for every graph $H$. Table~\ref{table-comparison} provides a comparative study of complexity (P vs NPC) between \tdset\ and \openpack\ in certain classes of graphs. The rest of the results in this article are motivated by Table~\ref{table-comparison} and the question, `Does there exist a graph class where these two problems `differ' from each other in the view of classical complexity (P vs NPC)?'
\begin{table}[h]
	\centering
	\captionsetup{format=hang}
	\caption{Comparison between \tdset\ and \openpack}
	\begin{tabular}{ l c c}
		\hline\\[-7pt]
		Graph Classes~~~~~~~~~~~~~& \tdset\ ~~~& \openpack\ \\
		\hline\\[-6pt]
		Chordal bipartite graphs~~~ & P~\cite{3K1,KRATSCH2000} &P~\cite{publish}\\
		$P_4$-free graphs & P~[Folklore] & P~[Folklore]\\
		Bipartite graphs & NPC~\cite{npc-tdset-bigraphs1983} & NPC~\cite{publish,Shalu17}\\
		Split Graphs & NPC~\cite{CORNEIL1984}& NPC~\cite{Igor2014}\\
		$K_{1,3}$-free Graphs& NPC~\cite{McRae1995} & NPC[*]\\
		\hline
	\end{tabular}
	\label{table-comparison}
\end{table}   \\
\noindent We answered this question in the affirmative (see Table~\ref{table-compar-diff}) with the following list of dichotomy results in subclasses of split graphs.
\begin{enumerate}
	\item \openpack\ is NP-complete on $K_{1,r}$-free split graphs for $r\geq 4$, and is polynomial time solvable for $r\leq 3$.
	\item \openpack\ is NP-complete on $I_r$-split graphs (see Section~\ref{sec:prelim} for definition) for $r\geq 3$ and is polynomial time solvable for $r\leq 2$.
	\item \tdset\ is NP-complete in $I_r$-split graphs for $r\geq 2$ and is polynomial time solvable when $r=1$ (a minor modification in the reduction by Corneil and Perl~\cite{CORNEIL1984} will also prove that \tdset\ is NP-complete for $I_2$-split graphs). 
\end{enumerate}
\begin{table}[h]
	\centering
	\captionsetup{format=hang}
	\caption{Complexity `difference' between \tdset\ and \openpack\ in subclasses of split graphs}
	\begin{tabular}{ l c c}
		\hline\\[-7pt]
		Graph Classes~~~~~~~~~~~~& \tdset\ ~~~& \openpack\ \\
		\hline\\[-6pt]
		$I_2$-Split Graphs & NPC \cite{CORNEIL1984} & P[*]\footnotemark\\
		$K_{1,4}$-free Split Graphs& P~\cite{RENJITH2020246} & NPC[*]\\
		\hline
	\end{tabular}
	\label{table-compar-diff}
\end{table}
\section{Preliminaries}\label{sec:prelim}\footnotetext{[*] denotes the results in this work}
 We follow West~\cite{west} for terminology and notation. The graphs considered in this work are simple and undirected unless specified otherwise. Given a graph $G(V,E)$, let $n$ and $m$ denote the number of vertices and the number of edges in $G$, respectively. Given a vertex $x\in V(G)$, the (open) neighbourhood of $x$ in $G$ is defined as $N_G(x)=\{y\in V(G)\,:\, xy\in E(G)\}$, and let the degree of a vertex $x$ in $G$ be $deg_G(x)=|N_G(x)|$. The closed neighbourhood of a vertex $x$ in $G$ is defined as $N_G[x]=\{x\}\cup N_G(x)$. We use $N(x)$, $N[x]$ and $deg(x)$ instead of $N_G(x)$, $N_G[x]$ and $deg_G(v)$ respectively, when there is no ambiguity on $G$. A vertex $x$ in $G$ is called an isolated vertex in $G$, if $N_G(x)=\emptyset$. Given $U\subseteq V(G)$, the subgraph of $G$ induced by $U$ is denoted as $G[U]$. Given a graph $H$, $G$ is said to be $H$-free if no induced subgraph of $G$ is isomorphic to $H$. For a vertex $x\in V(G)$, let $E_G(x)$ denote the set of all edges incident on $x$, and for an edge $e\in E(G)$, let $V_G(e)$ denote the end vertices of $e$ in $G$. Note that for $u\in V(G)$ and $e\in E(G)$, the edge $e\in E_G(u)$ if and only if $u\in V_G(e)$. For a graph $G$, the line graph $L(G)$ of $G$ is a graph with $V(L(G))=E(G)$ and two elements $e,e'\in V(L(G))$ are adjacent in $L(G)$ if and only if $V_G(e)\cap V_G(e')\neq \emptyset$. Given two (not necessarily disjoint) graphs $H$ and $H'$, the graph union $H\cup H'$ is defined as $V(H\cup H')=V(H)\cup V(H')$ and $E(H\cup H')=E(H)\cup E(H')$. If $V(H)$ and $V(H')$ are not disjoint, then we call the graph union $H\cup H'$ as {\it merger} of the graphs $H$ and $H'$. For $p\in \mathbb{N}\cup \{0\}$ and a graph $H$, the graph $pH$ is defined as the union of $p$ disjoint copies of $H$. Given two disjoint graphs $G$ and $H$, we say that (i) a vertex $x$ of the graph $G$ is {\it replaced} by the graph $H$ if a graph $G'$ is constructed with $V(G')=(V(G)\setminus \{x\})\cup V(H)$ and $E(G')=(E(G)\setminus \{xy:y\in V(G)\text{ and }xy\in E(G)\})\cup E(H)\cup \{wy:w\in V(H), y\in (V(G)\setminus\{x\})\text{ and }xy\in E(G)\}$, and (ii) a vertex $x$ of the graph $G$ is \textit{identified} with a vertex $y$ of the graph $H$, if the vertex $y$ in $H$ is relabelled as $x$ and the graphs $G$ and $H$ are merged.\\
\noindent A vertex subset $C$ of $G$ is called a {\it clique} in $G$ if every pair of distinct vertices in $C$ is adjacent in $G$. A vertex subset $I$ of $G$ is called an {\it independent set} in $G$ if no pair of vertices in $I$ are adjacent in $G$. The size of a largest independent set in $G$ is called the {\it independence number} of $G$, and is denoted by $\alpha(G)$. Let $P_n$, $C_n$ and $K_n$ denote the path, cycle and complete graph on $n$ vertices, respectively. A set $D\subseteq V(G)$ is called a {\it dominating set} in $G$, if every vertex in $V(G)\setminus D$ is adjacent to some vertex in $D$.  An edge subset $F$ of a graph $G$ is called a {\it matching} in $G$, if $V_G(e)\cap V_G(e')=\emptyset$ for every pair of distinct $e,e'\in F$. The cardinality of a largest matching in $G$ is called the matching number, $\alpha'(G)$, of $G$.\\
A graph $G(V,E)$ is said to be a {\it split graph}, if there exists a partition $V(G)=C\cup I$ such that $C$ is a clique and $I$ is an independent set in $G$ and is denoted as $G(C\cup I, E)$. Note that for $r\geq 2$, if a split graph $G$ is $K_{1,r}$-free, then $|N_G(v)\cap I|\leq r-1$ for every $v\in C$. In accordance with this observation, the class of {\it $I_r$-split graphs} is defined for every natural number $r$ as a split graph $G(C\cup I, E)$ with $deg_G(v)=r$ for every $v\in I$.\\
 Given a graph $G$ and a positive integer $k\leq |V(G)|$, the problem \indset\ asks whether $G$ has an independent set of size at least $k$. Given a graph $G$, the goal of the problem \maxindset\ is to find $\alpha(G)$. For $r\in \mathbb{N}$, given a collection of sets $X_1,X_2,\ldots,X_r$ each of cardinality $q$ for some $q\in \mathbb{N}$ and a subset $M$ of $\prod\limits_{i=1}^rX_i$, the $r$-\textsc{Dimensional Matching} problem asks whether there exists $L\subseteq M$ such that (i) $|L|=q$ and (ii) for every pair of $r$-tuple $x=(x_1,x_2,\ldots,x_r)$ and $y=(y_1,y_2,\dots,y_r)$ in $L$, $x_i\neq y_i$ for $1\leq i\leq r$. For $r\in \mathbb{N}$, given a non-empty set $U$, a set $\mathcal{W}$ of $r$-sized subsets of $U$ and a positive integer $k$, the $r$-\textsc{Hitting Set} problem asks whether there exists a subset $X$ of $U$ such that $|X|\leq k$ and $X\cap W\neq \emptyset$ for every $W\in \mathcal{W}$. \\
 Also, we refer the reader to \cite{param-saket,parameterized-downey-2013} for a brief note on parameterized algorithms, intractability and W-hierarchy. 
\section{\boldmath $H$-free Graphs}\label{sec:h-free}
We dedicate this section to prove the dichotomy result on \openpack\ stated in Theorem~\ref{thm-h-dic}. Observation~\ref{obs-h-free} helps us prove the necessary part of Theorem~\ref{thm-h-dic}. 
\begin{observation}
	For $p\geq 4$, let $H$ be a graph on $p$ vertices such that $H\notin\{P_4\cup (p-4)K_1,P_3\cup (p-3)K_1, K_2\cup (p-2)K_1, pK_1\}$. Then, $H$ contains at least one of $K_3,2K_2,C_4,K_{1,3}$, or $C_5$ as an induced subgraph.
	\label{obs-h-free}
\end{observation}
\begin{proof}
	On the contrary, assume that $H\notin\{P_4\cup (p-4)K_1,P_3\cup (p-3)K_1, K_2\cup (p-2)K_1, pK_1\}$ and $H$ is a $\{K_3,2K_2,C_4,K_{1,3},C_5\}$-free graph. Then, the following statements hold.
	\begin{enumerate}
		\item[(1)] $H$ has exactly one non-trivial component since $H$ is $2K_2$-free and $H\ncong pK_1$.
		\item[(2)] For every vertex $v\in V(H)$, $deg_H(v)\leq2$ since $H$ is $\{K_3,K_{1,3}\}$-free.
		\item[(3)] $H$ is acyclic since $H$ is $\{K_3,2K_2,C_4,C_5\}$-free.
	\end{enumerate}
\noindent	Note that by (1), (2), and (3), it is clear that the non-trivial component $H'$ of $H$ is a path and since $H$ is $2K_2$-free, $H'\in \{K_2,P_3,P_4\}$. This implies that $H\in \{P_4\cup (p-4)K_1,P_3\cup (p-3)K_1, K_2\cup (p-2)K_1\}$, a contradiction.
\end{proof}
\begin{remark}
	\openpack\ is NP-complete on (i) $2K_2$-free graphs, (ii) $C_4$-free graphs, (iii) $C_5$-free graphs, and (iv) triangle-free graphs because of the facts that \openpack\ is NP-complete on (a) split graphs (i.e., the class of $\{2K_2,C_4,C_5\}$-free graphs)~\cite{Igor2014} and (b) bipartite graphs (a subclass of $K_3$-free graphs) \cite{publish,Shalu17} and (c) \openpack\ is in NP in every subclass of simple graphs.	\label{rem-triangle}
\end{remark}
\noindent Next, we prove the following results on $K_{1,3}$-free graphs: (i) \openpack\ is NP-complete, (ii) \maxopenpack\ is hard to approximate within a factor of $(N^{(\frac{1}{2}-\delta)})$ for any $\delta>0$ unless P = NP (where $N$ denotes the number of vertices in a $K_{1,3}$-free graph), and (iii) \openpack\ parameterized by solution size is W[1]-complete.
\subsection{\boldmath $K_{1,3}$-free Graphs}\label{sec:k_1,3}
The following construction gives a polynomial time reduction from \indset\ on simple graphs to \openpack\ on $K_{1,3}$-free graphs, where the former problem is known to be NP-complete~\cite{karp}.
\begin{construct}
	\emph{Input:} A simple graph $G$ with $V(G)=\{u_1,u_2,\ldots,u_n\}$.\\
	\emph{Output:} A $K_{1,3}$-free graph $G'$.\\
	\emph{Guarantee:} $G$ has an independent set of size $k$ if and only if $G'$ has an open packing of size $k$.\\
	\emph{Procedure:}~\\
	\begin{tabular}{l >{\RaggedRight}p{13.25cm}}
		Step 1: & Replace each edge $e=uu'$ in $G$ with a three vertex path $ueu'$ in $G'$.\\
		Step 2: &For every $u\in V(G)$, make $E_G(u)$ a clique in $G'$.\\
		Step 3: &  For every vertex $u_i\in V(G)$ with exactly one edge, say $e$ incident on it in $G$, introduce a vertex $v_i$, and two edges $u_iv_i$ and $v_ie$ in $G'$.
	\end{tabular}
	\label{construct-op-id-k_1,3}
\end{construct}
\noindent An example of Construction~\ref{construct-op-id-k_1,3} is given in Fig.~\ref{fig:k1,3-op-id}. Also, note that Step 2 of Construction~\ref{construct-op-id-k_1,3} can be viewed as the merger of the line graph of $G$ and the graph obtained in Step 1. Further, $V(G')=V(G)\cup E(G)\cup A$, where $A=\{v_i\,:\, 1\leq i\leq n\text{ and } deg_G(u_i)=1\}$ and $E(G')=E_1\cup E_2\cup E_3\cup E_4$, where $E_1=\{ue\,:\,  u\in V(G), e\in E(G)\text{ and } e\in E_G(u)\}$, $E_2= E(L(G))$, $E_3=\{v_iu_i\,:\, 1\leq i\leq n\text{ and } deg_G(u_i)=1\}$ and $E_4= \{v_ie\,:\, 1\leq i\leq n \text{ and } E_G(u_i)=\{e\} \}$. Hence, $|V(G')|\leq  2n+m$ and $|E(G')|=\sum_{i=1}^4|E_i|\leq 2m+\binom{m}{2}+2n$. Thus, the graph $G'$ can be constructed in quadratic time in the input size. 
\begin{figure}
	\centering
	\begin{tikzpicture}[scale=.85]
		\draw (2,1)node[block1,label=right:$u_2$]{}--(2,3.5)node(u3)[block1,label=right:$u_3$]{}--(0,5)node(u4)[block1,label=below:$u_4$]{}--(-2,3.5)node(u5)[block1,label=left:$u_5$]{}--(-2,1)node(u1)[block1,label=left:$u_1$]{}--cycle;
		\draw (u4)--(0,6.5)node(u6)[block1,label=left:$u_6$]{};
		\draw (0.25,5.65)node[label=left:$e_7$]{};
		\draw (u3)--(u5);
		\draw (0,-.5)node{(a)\,$G$};
		\draw (0,1.25)node{$e_1$};
		\draw (1.7,2.25)node{$e_2$};
		\draw (-1.7,2.25)node{$e_6$};
		\draw (1.5,4.25)node{$e_4$};
		\draw (-1.5,4.25)node{$e_5$};
		\draw (0,3)node{$e_3$};
		
		\draw (6,0.75)node(v1)[block1,label=left:$u_1$]{}--(8,0.75)node(e1)[block1,label=below:$e_1$]{}--(10,0.75)node[block1,label=right:$u_2$]{}--(10,2.15)node(e2)[block1,label=left:$e_2$]{}--(10,3.5)node(v3)[block1,label=right:$u_3$]{}--(9.1,4.25)node(e4)[block1,label=right:$e_4$]{}--(8,5.25)node(v4)[block1,label=below:$u_4$]{}--(6.9,4.25)node(e5)[block1,label=left:$e_5$]{}--(6,3.5)node(v5)[block1,label=left:$u_5$]{}--(6,2.1)node(e6)[block1,label=right:$e_6$]{}--cycle;
		\draw (v5)--(8,3.5)node(e3)[block1,label=below:$e_3$]{}--(v3);
		\draw (v4)--(8,6)node(e7)[block1,label=above:$e_7$]{}--(6.9,7)node(v6)[block1,label=left:$u_6$]{}--(9.1,7)node(w6)[block1,label=right:$v_6$]{}--(e7)--(e5);
		\draw (e7)--(e4);
		\draw (e3)--(e4)--(e5)--(e3)--(e6)--(e1)--(e2)--(e3);
		\draw (e2)--(e4);
		\draw (e5)--(e6);
		\draw (8,-.5)node{(b)\,$G'$};
	\end{tikzpicture}
	\captionsetup{format=hang}
	\caption{(a) a simple graph $G$, (b) $K_{1,3}$-free graph $G'$ obtained from $G$ using Construction~\ref{construct-op-id-k_1,3}.}
	\label{fig:k1,3-op-id}
\end{figure}
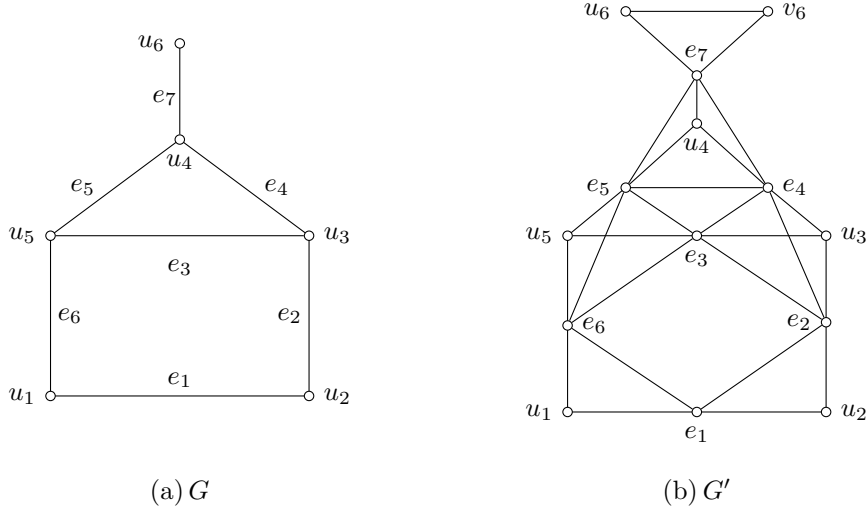\\
\noindent We can observe that in the graph $G'$, $\text{for every vertex }u_i \in V(G)$,
\begin{equation}
N_{G'}(u_i)=\begin{cases}
	E_G(u_i)\cup \{v_i\}~~~\text{ if }|E_G(u_i)|=deg_G(u_i)=1\\
	E_G(u_i)\text{\qquad\quad~~ otherwise}
\end{cases}
\label{equ-k-1,3-a}
\end{equation} 
\begin{equation}
 N_{G'}(v_i)=\{u_i,e\} \text{ if }E_G(u_i)=\{e\} \text{ i.e., }deg_G(u_i)=1
 \label{equ-k-1,3-b}
\end{equation} The above equations imply that, by Steps 2 and 3 of Construction~\ref{construct-op-id-k_1,3}, $N_{G'}(u_i)$ and $N_{G'}(v_i)$ are cliques in $G'$. Also, by Construction~\ref{construct-op-id-k_1,3}, for every $e\in E(G)$ with $V_G(e)=\{u,u'\}$, \begin{equation}
N_{G'}[e]=N_{G'}[u]\cup N_{G'}[u'] \text{ and so } N_{G'}(e)=(N_{G'}[u]\cup N_{G'}[u'])\setminus\{e\}
\label{equ-k-1,3-c}
\end{equation} This implies that $N_{G'}(e)$ is a union of two cliques in $G'$. We know that a graph $H$ is $K_{1,3}$-free if and only if $N_H(x)$ is a union of at most two cliques for all $x\in V(H)$. So, Equations~(\ref{equ-k-1,3-a}), (\ref{equ-k-1,3-b}), and (\ref{equ-k-1,3-c}) imply that $G'$ is $K_{1,3}$-free.\\ \noindent Next, we prove Claims~\ref{claim-s-iff-k_1,3}-\ref{claim-op-to-ind-k-1,3}, which are helpful in completing the proof of guarantee of Construction~\ref{construct-op-id-k_1,3}. 
\begin{claim-alt}
A set $S\subseteq V(G)$ is an independent set in the input graph $G$ of Construction~\ref{construct-op-id-k_1,3} if and only if $S$ is an open packing in the output graph $G'$.\\
	Note that by Equation~(\ref{equ-k-1,3-a}), for $i\neq j$, $e=u_iu_j$ in $G$ if and only if $N_{G'}(u_i)\cap N_{G'}(u_j)=\{e\}$, and hence $u_iu_j$ is not an edge in $G$ if and only if $N_{G'}(u_i)\cap N_{G'}(u_j)=\emptyset$. This completes the proof of Claim~\ref{claim-s-iff-k_1,3}. 
	\label{claim-s-iff-k_1,3}
 \end{claim-alt}
 \begin{observation}
 	Let $H$ be a graph, and let $U=\{x,y,z\}$ be a vertex subset of $H$ that induces a triangle in $H$. Then, for any open packing $S$ in $H$, $|S\cap U|\leq 1$. 
 	\label{obs-added}
 \end{observation}
 \noindent The above observation is held by the fact that any two vertices in a triangle have a common neighbour (third vertex).
 \begin{claim-alt}
	Let $u_i\in V(G)$ such that $deg_G(u_i)=1$, and let $S$ be an open packing in $G'$ with $v_i\in S$. Then, (i) $u_i\notin S$, (ii) $S'=(S\setminus\{v_i\})\cup \{u_i\}$ is an open packing in $G'$, and (iii) $|S|=|S'|$.\\
	Since $deg_G(u_i)=1$, there exists an edge $e\in E(G)$ such that $E_G(u_i)=\{e\}$. Then, $\{u_i,v_i,e\}$ induces a triangle in $G'$ with $N_{G'}(u_i)=\{e,v_i\}$ and $N_{G'}(v_i)=\{e,u_i\}$. Therefore, by Observation~\ref{obs-added}, for any open packing $S_1$ in $G'$, $|S_1\cap \{u_i,v_i,e\}|\leq 1$. Thus, $v_i\in S$ implies that $u_i,e\notin S$. This proves (i) of Claim~\ref{claim-abc}. Next, we prove (ii). Since $N_{G'}(u_i)=\{e,v_i\}$ and $N_{G'}(v_i)=\{e,u_i\}$, for any vertex $x\in S\setminus \{v_i\}$, $x$ and $v_i$ have no common neighbour in $G'$, which implies that $x$ and $u_i$ have no common neighbour in $G'$. This, together with the fact that $S\setminus\{v_i\}$ is an open packing in $G'$, implies that $S'=(S\setminus\{v_i\})\cup \{u_i\}$ is an open packing in $G'$. Note that (iii) follows from (i).\label{claim-abc}
 \end{claim-alt}
 \begin{claim-alt}
 	Let $e$ be an edge incident on a vertex $u_i$ in $G$. Then, for any open packing $S$ in $G'$, $|S\cap \{e,u_i\}|\leq 1$.\\
 	If $deg_G(u_i)=1$, then $\{u_i,v_i,e\}$ induces a triangle in $G'$ and so $|S\cap \{e,u_i\}|\leq 1$ by Observation~\ref{obs-added}. Also, if $deg_G(u_i)\geq 2$, then $\{u_i,e,f\}$ induces a triangle in $G'$ for some $f\in (E_G(u_i)\setminus \{e\})$. Again by Observation~\ref{obs-added}, $|S\cap \{u_i,e\}|\leq 1$. \label{claim-added-2}
 \end{claim-alt}
 \begin{claim-alt}
	Let $S$ be an open packing in $G'$ with $e=u_\ell u_j\in S\cap E(G)$ for some $1\leq \ell<j\leq n$. Then, (i) $u_\ell,u_j\notin S$, (ii) $S_\ell=(S\setminus\{e\})\cup \{u_\ell\}$ and $S_{j}=(S\setminus\{e\})\cup \{u_j\}$ are open packings in $G'$, and (iii) $|S|=|S_\ell|=|S_{j}|$.\\
	Since $e\in S$, $u_\ell,u_j\notin S$ because of Claim~\ref{claim-added-2}. This proves (i). Also, if $v_\ell \in V(G')$, then $v_\ell \notin S$ because of Observation~\ref{obs-added} and the fact that $\{v_\ell, e,u_\ell\}$ induces a triangle in $G'$. Similarly, if $v_j\in V(G')$, then $v_j\notin S$. Next, we prove that $S_\ell =(S\setminus \{e\})\cup \{u_\ell\}$ is an open packing in $G'$. Note that $S\setminus \{e\}$ is an open packing in $G'$. So, if $S_\ell$ is not an open packing in $G$, then there exists a vertex $z\in S\setminus\{e\}$ such that $N_{G'}(z)\cap N_{G'}(u_\ell)\neq \emptyset$. Let $y \in N_{G'}(z)\cap N_{G'}(u_\ell)$. Note that $zu_\ell \notin E(G')$, else $u_\ell \in N_{G'}(z)\cap N_{G'}(e)$, which would be a contradiction to $e,z \in S$ (an open packing in $G'$). Similarly, $zu_j\notin E(G')$. If $y=e$, then since $e=y\in N_{G'}(z)$, we have $z\in (N_{G'}(e)\setminus (N_{G'}[u_j]\cup N_{G'}[u_\ell]))$, a contradiction to Equation~(\ref{equ-k-1,3-c}) (see Fig.~\ref{fig-extras-k1,3}). Therefore, $y\neq e$. This implies that $ye\in E(G')$, since $y,e \in N_{G'}(u_\ell)$, which is a clique in $G'$. Then, $y\in N_{G'}(z)\cap N_{G'}(e)$, a contradiction to $S$ being an open packing in $G'$. Therefore, $S_\ell =(S\setminus \{e\})\cup \{u_\ell\}$ is an open packing in $G'$. A similar argument holds for $S_j$. Also, note that (iii) follows from (i).
	\label{claim-eu}
 \end{claim-alt}
 \begin{figure}
 	\centering
 	\begin{tikzpicture}
 		\draw (0,0)node[block1,label=right:$u_\ell$]{}--(0,2)node[block1,label=right:$u_j$]{};
 		\draw (0,1)node[label=left:$e$]{};
 		\draw (4,0)node(ul)[block1,label=below:$u_\ell$]{}--(4,1)node(e)[block1,label=left:$e$]{}--(5,2)node(uj)[block1,label=above:$u_j$]{};
 		\draw (ul)--(6,0)node(y)[block1,label=below:$y$]{}--(6,1)node(z)[block1,label=right:$z$]{};
 		\draw[dashed] (e)--(z)--(ul);
 		\draw[dashed] (uj)--(z);
 		\draw[thick] (e)--(y);
 		\draw (0,-1)node{(a)};
 		\draw (5,-1)node{(b)};
 	\end{tikzpicture}
 	\captionsetup{format=hang}
 	\caption{(a) an edge $e=u_\ell u_j$ of the input graph $G$ of Construction~\ref{construct-op-id-k_1,3} and (b) the three vertex path $u_\ell eu_j$ produced by Construction~\ref{construct-op-id-k_1,3} along with two vertices $y,z$ of $G'$ such that $y\in N_{G'}(z)\cap N_{G'}(u_\ell)$ as in the contrary assumption of Claim~\ref{claim-eu}. In figure, dashed lines represent the edges not in $G'$ as observed in the proof of Claim~\ref{claim-eu}.}
 	\label{fig-extras-k1,3}
 \end{figure}
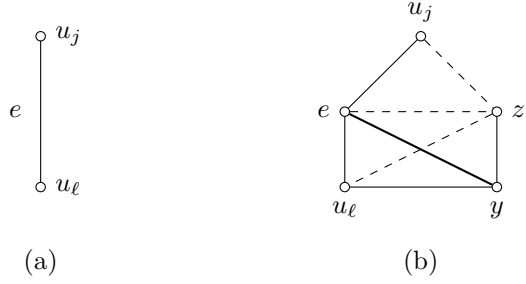
\begin{claim-alt}
	Let $S$ be an open packing of size $k$ in $G'$. Then, there exists an independent set $S^*$ of size $k$ in $G$.\\
	Note that $V(G')=V(G)\cup E(G)\cup A$. We construct $S^*$ from $S$ in two phases. First, we construct an open packing $S_1$ in $G'$ such that $S_1\subseteq V(G)\cup E(G)$ and $|S_1|=|S|$. If $S\cap A=\emptyset$, then $S_1:=S$. Else $S\cap A\neq \emptyset$ and let $S\cap A=\{v_{i_1},v_{i_2},\ldots,v_{i_p}\}$. Then, $S_1:=(S\setminus \{v_{i_1},v_{i_2},\ldots,v_{i_p}\})\cup \{u_{i_1},u_{i_2},\ldots,u_{i_p}\}$. By repeated application of Claim~\ref{claim-abc}, $S_1$ is an open packing in $G'$ such that $|S_1|=|S|$ and $S_1\subseteq V(G)\cup E(G)$. Next, we construct an open packing $S^*$ of $G'$ from $S_1$ such that $|S^*|=|S_1|$ and $S^*\subseteq V(G)$. If $S_1\cap E(G)=\emptyset$, then $S^*:=S_1$. If not, let $S_1\cap E(G)=\{e_1,e_2,\ldots,e_q\}$, and let $u_{\ell_i}$ be an end vertex of $e_i$ in $G$ for $1\leq i\leq q$. Note that for $1\leq i<j\leq q$, $u_{\ell_i}\neq u_{\ell_j}$; otherwise, $\{u_{\ell_i},e_i,e_j\}$ induces a triangle (by Equation~(\ref{equ-k-1,3-a})) in $G'$ with $e_i,e_j\in S_1$ (an open packing in $G'$), a contradiction to Observation~\ref{obs-added}. Thus, $S^*:=(S_1\setminus \{e_1,e_2,\ldots,e_q\})\cup \{u_{\ell_1},u_{\ell_2},\ldots,u_{\ell_q}\}$. By repeated application of Claim~\ref{claim-eu}, $S^*$ is an open packing in $G'$ with $|S^*|=|S_1|=|S|=k$ and $S^*\subseteq V(G)$. By Claim~\ref{claim-s-iff-k_1,3}, $S^*$ is an independent set in $G$. This proves Claim~\ref{claim-op-to-ind-k-1,3}. \label{claim-op-to-ind-k-1,3}
\end{claim-alt}
\noindent Claims~\ref{claim-s-iff-k_1,3} and \ref{claim-op-to-ind-k-1,3} prove the guarantee of Construction~\ref{construct-op-id-k_1,3}. Thus, we have the following theorem which is an implication of Construction~\ref{construct-op-id-k_1,3} and the fact that \indset\ is NP-complete on simple graphs~\cite{karp}.
\begin{theorem}
	\openpack\ is NP-complete on $K_{1,3}$-free graphs.
	\label{thm-k-1,3-npc}
\end{theorem}	
\noindent  We use the theorem below to derive an in-approximation result of \maxopenpack\ for $K_{1,3}$-free graphs.
\begin{theorem}[\cite{hastard}]
	\maxindset\ cannot be \hypertarget{subapprx}{approximated} within a factor of $n^{(1-\epsilon)}$ for any $\epsilon>0$, in general graphs unless P = NP.
	\label{hastard-ind-inapprx}
\end{theorem}
\begin{theorem}
	\maxopenpack\ is hard to approximate within a factor of $N^{\frac{1}{2}-\delta}$ for any $\delta>0$ in $K_{1,3}$-free graphs unless P = NP, where $N$ denotes the number of vertices in a $K_{1,3}$-free graph.
	\label{thm-k_1,3-inapprx}
\end{theorem}
\begin{proof}
	On the contrary, assume that \maxopenpack\ admits a $\displaystyle N^{\left(\frac{1}{2}-\delta\right)}$-factor approximation algorithm for some $\delta>0$ in $K_{1,3}$-free graphs. Then, an open packing $S$ of the graph $G'$ constructed in Construction~\ref{construct-op-id-k_1,3} with $\displaystyle |S|> \frac{\rho^o(G)}{|V(G')|^{\left(\frac{1}{2}-\delta\right)}}$ can be found in polynomial time. Hence, an independent set $U$ of the input graph $G$ in Construction~\ref{construct-op-id-k_1,3} with $|U|=|S|$ can be found in polynomial time. Note that $N=|V(G')|\leq 2n+m\leq 2n+\frac{n^2}{2}\leq n^2$ for $n\geq 4$. Then,
	\begin{align*}
		|S|&>\frac{\rho^o(G')}{N^{\left(\frac{1}{2}-\delta\right)}}\hspace{2cm}\\
		|S|&> \frac{\rho^o(G')}{(n^2)^{\left(\frac{1}{2}-\delta\right)}} \hspace{5mm} \text{for }n\geq 4\\
		|S|&>\frac{\alpha(G)}{n^{1-2\delta}}\\
		|U|&>\frac{\alpha(G)}{n^{1-2\delta}}
	\end{align*}
	Since $\rho^o(G')=\alpha(G)$ by the guarantee of Construction~\ref{construct-op-id-k_1,3}.\\
	Let $\epsilon=2\delta$.
	Then, $\epsilon>0$ and for $n\geq 4$, we have
	\begin{align*}
		|U|&>\frac{\alpha(G)}{n^{1-\epsilon}}
	\end{align*}
	The above inequality implies that \maxindset\ has a $(n^{1-\epsilon})$-factor approximation algorithm for some $\epsilon>0$, which contradicts Theorem~\ref{hastard-ind-inapprx}.
\end{proof} 
\noindent Next, we give a parameterized intractability result for \openpack\ on $K_{1,3}$-free split graphs using Construction~\ref{construct-op-id-k_1,3} and the following result from the literature.
\begin{lemma}[\cite{Rall2005}]
	Given a graph $G$, let the {\it neighbourhood graph} $G^{[o]}$ of $G$ be a simple graph with $V(G^{[o]})=V(G)$ and $E(G^{[o]})=\{xy\,:\, x,y\in V(G), x\neq y \text{ and } N_G(x)\cap N_G(y)\neq \emptyset\}$. Then, a vertex subset $S$ is an open packing in $G$ if and only if $S$ is an independent set in $G^{[o]}$.
	\label{op-in-g-iff-ind-in-go}
\end{lemma}
\begin{theorem}
	\openpack\ parameterized by solution size is W[1]-complete on $K_{1,3}$-free graphs. \label{thm-k_1,3-w1}
\end{theorem}
\begin{proof} Note that \indset\ parameterized by solution size is W[1]-complete~\cite{W1-completeness-1995}. Lemma~\ref{op-in-g-iff-ind-in-go} implies that every instance $(G,k)$ of \openpack\ on $K_{1,3}$-free graphs can be (FPT) reduced into an instance $(G^{[o]},k)$ of \indset. Hence, \openpack\ parameterized by solution size on $K_{1,3}$-free graphs is in W[1]. The guarantee of Construction~\ref{construct-op-id-k_1,3} implies that every instance $(G,k)$ of \indset\ can be (FPT) reduced into an instance $(G',k)$ of \openpack\ on $K_{1,3}$-free graphs. Thus, \openpack\ parameterized by solution size is W[1]-hard on $K_{1,3}$-free graphs. 
\end{proof}
\subsection{\boldmath Subclasses of $(P_4\cup rK_1)$-free Graphs}\label{sec:sup-rk1}
In this section, we prove bounds on the total domination number and the open packing number in the classes of (i) $(P_4\cup rK_1)$-free graphs, (ii) $(P_3\cup rK_1)$-free graphs, (iii) $(K_2\cup rK_1)$-free graphs, and (iv) $rK_1$-free graphs. These bounds eventually solve \tdset\ and \openpack\ in polynomial time in the graph classes mentioned above. The following lemmas are the tools used to connect the bounds with polynomial time algorithms.
\begin{lemma}[Folklore]
	Given a graph $G$ and a vertex subset $D$ of $G$, it can be tested whether $D$ is a total dominating set of $G$ or not in $O(n+m)$ time.
	\label{td-testing}
\end{lemma}
\begin{lemma}
	Given a graph class $\mathcal{G}$, if there exists $\ell \in \mathbb{N}$ such that $\gamma_t(G)\leq \ell$ for every $G\in \mathcal{G}$, then an optimal total dominating set of $G$ can be found in $O(n^\ell (n+m))$ time for every $G\in \mathcal{G}$.
	\label{lem-td-opt-ell}
\end{lemma}
\begin{proof}
	Let $G\in \mathcal{G}$, and let $\mathcal{D}_\ell$ be the set of all total dominating sets of size at most $\ell$ in $G$. Also, let $D$ be an optimal total dominating set in $G$. Then, since $|D|=\gamma_t(G)\leq \ell$, $D\in \mathcal{D}_\ell$. Also, $\mathcal{D}_\ell\subseteq \mathcal{V}_\ell =\{X\subseteq V(G)\,:\, |X|\leq \ell\}$. Note that the collection $\mathcal{V}_\ell$ can be found in $O(n^\ell)$ time. Also, for every $X\in \mathcal{V}_\ell$, it can be tested whether $X\in \mathcal{D}_\ell$ in $O(n+m)$ time using Lemma~\ref{td-testing}. This, together with the fact that $|\mathcal{V}_\ell|=O(n^\ell)$, implies that the collection $\mathcal{D}_\ell$ can be found in $O(n^\ell (n+m))$ time. Since every $D'\in \mathcal{D}_\ell$ is a total dominating set, $\gamma_t(G)\leq \min\{|D'|\,:\,D'\in \mathcal{D}_\ell\}$. Further, since $D\in \mathcal{D}_\ell$, $\gamma_t(G)=|D|\geq \min\{|D'|\,:\,D'\in \mathcal{D}_\ell\}$. So, $\gamma_t(G)= \min\{|D'|\,:\,D'\in \mathcal{D}_\ell\}$. Since $\mathcal{D}_\ell$ can be generated in $O(n^\ell (n+m))$ time, so is $\gamma_t(G)$ and an optimal total dominating set of $G$. 
\end{proof}
\noindent The following lemma is helpful in proving a result similar to the above lemma for open packing.
\begin{lemma}[\cite{publish}]
	Given a graph $G$ and a vertex subset $S$ of $G$, it can be tested in $O(n)$ time whether $S$ is an open packing in $G$ or not. \label{thm-openpacking-testing}
\end{lemma}
\noindent Proof of Lemma~\ref{thm-openpacking-testing} is given in Appendix~\ref{app-sec:op-test}.
\begin{lemma}
	Given a graph class $\mathcal{G}$, if there exists $k\in \mathbb{N}$ such that $\rho^o(G)\leq k$ for every $G\in \mathcal{G}$, then (i) $G$ contains at most $O(n^k)$ open packings, and (ii) all open packings in $G$ can be computed in $O(n^{k+1})$ time for every $G\in \mathcal{G}$. So, $\rho^o(G)$ can be computed in $O(n^{k+1})$ time.
	\label{lem-tc-bound}
\end{lemma}
\begin{proof}
	Let $G\in \mathcal{G}$. Let $\mathcal{S}$ be the set of all open packings in $G$. We prove that the set $\mathcal{S}$ can be found in $O(n^{k+1})$ time. Let $\mathcal{V}_k=\{X\subseteq V(G)\,:\, |X|\leq k\}$. Then, since $\rho^o(G)\leq k$, $\mathcal{S}\subseteq \mathcal{V}_k$. Note that the set $\mathcal{V}_k$ is of size $O(n^k)$ and can be found in $O(n^k)$ time. Since $\mathcal{S}\subseteq \mathcal{V}_k$, $|\mathcal{S}|=O(n^k)$. Also, for every vertex subset $X\in \mathcal{V}_k$, it can be tested in $O(n)$ time whether $X\in \mathcal{S}$ using Lemma~\ref{thm-openpacking-testing}. Thus, the set $\mathcal{S}$ can be found in $O(n^{k+1})$ time. Note that $\rho^o(G)=\max\{|S|\,:\, S\in \mathcal{S}\}$. Since the set $\mathcal{S}$ can be found in $O(n^{k+1})$ time, $\rho^o(G)$ can be computed in $O(n^{k+1})$ time. 
\end{proof}
\noindent Next, we study the bounds on the total domination number and the open packing number in the class of $(P_t\cup rK_1)$-free graphs for $t\in \{1,2,3,4\}$ and every non-negative integer $r\geq 3-t$.
\subsubsection*{\boldmath $(P_4\cup rK_1)$-free Graphs}
In this section, we prove that for $r\geq 1$,  $\gamma_t(G)\leq 2r+2$, and $\rho^o(G)\leq 2r+1$ in a connected ($P_4\cup rK_1$)-free graph $G$ (the bound on the open packing number of $(P_4\cup rK_1)$-free graphs along with Lemma~\ref{lem-tc-bound} will solve the sufficiency part of Theorem~\ref{thm-h-dic}). Further, we show that the bound on the open packing number in these graph classes is tight (see Remark~\ref{rem-tight-p4-rk1}) and that the bound on the total domination number is tight for the class of $(P_4\cup K_1)$-free graphs (see Remark~\ref{rem-td-p4-k1-tight}). The following known lemma solves the problems for $P_4$-free graphs.
\begin{lemma}[Folklore]
	Let $G$ be a connected $P_4$-free graph. Then, $\rho^o(G)\leq \gamma_t(G)= 2$.
	\label{p4-td-bound}
\end{lemma}
\noindent The lemma below shows that for $r\geq 1$, the total domination number of a connected $(P_4\cup rK_1)$-free graph is bounded above by a function of $r$.
\begin{lemma}
	For $r\geq 1$, if $G$ is a connected ($P_4\cup rK_1$)-free graph, then $\gamma_t(G)\leq 2r+2$.	\label{p4rk1-td-bound}
\end{lemma}
\begin{proof}
	We prove this lemma by induction on $r\geq 0$. Note that Lemma~\ref{p4-td-bound} shows that the bound holds for the case when $r=0$. So, assume that for $r\geq 1$, $\gamma_t(G')\leq 2(r-1)+2$ if $G'$ is a connected ($P_4\cup (r-1)K_1$)-free graph. Next, we show that for every connected ($P_4\cup rK_1$)-free graph $G$, $\gamma_t(G)\leq 2r+2$. Let $G$ be a connected ($P_4\cup rK_1$)-free graph. Suppose that $G$ is ($P_4\cup (r-1)K_1$)-free, then by induction assumption $\gamma_t(G)\leq 2(r-1)+2<2r+2$. So, assume that $G$ contains a $P_4\cup (r-1)K_1$ as an induced subgraph. Let $D=\{x_1,x_2,x_3,x_4,y_1,y_2,\ldots,y_{r-1}\}$ be a vertex subset of $G$ such that $x_ix_{i+1}\in E(G)$ for $i=1,2,3$ and $G[D]\cong P_4\cup (r-1)K_1$ (see Fig.~\ref{fig:draw-p4-r-1k1}). Since $G$ is ($P_4\cup rK_1$)-free, every vertex in $V(G)\setminus D$ is adjacent to some vertex in $D$. Also, since $G$ is connected, for every $j=1,2,\ldots,(r-1)$, there exists a vertex $w_j\in V(G)$ such that $y_jw_j\in E(G)$. Also, the vertices $x_1,x_2,x_3$ and $x_4$ are adjacent to a vertex in $\{x_2,x_3\}\subseteq D$. Thus, every vertex in $V(G)=D\cup (V(G)\setminus D)$ is adjacent to some vertex in $\{w_j\,:\, 1\leq j\leq r-1\}\cup D$. So, $\{w_j\,:\, 1\leq j\leq r-1\}\cup D$ is a total dominating set in $G$, and hence $\gamma_t(G)\leq r-1 +|D|\leq 2(r-1)+4=2r+2$. 
\end{proof}
\begin{figure}
	\centering
	\begin{tikzpicture}[scale=0.95]
		\draw (0,0)node[block1,label=below:$x_1$]{}--(1,0)node[block1,label=below:$x_2$]{}--(2,0)node[block1,label=below:$x_3$]{}--(3,0)node[block1,label=below:$x_4$]{};
		\draw (4,0)node(y1)[block1,label=below:$y_1$]{};
		\draw (5,0)node(y2)[block1,label=below:$y_2$]{};
		\draw[thick] (5.75,0)node{$\cdots$};
		\draw (6.5,0)node(yr-1)[block1,label=below:$y_{r-1}$]{};
		\draw (3.25,-.10)ellipse[x radius= 4, y radius=.8];
		\draw (3.25,3)ellipse[x radius=4, y radius=.8];
		\draw (7.75,0)node{$D$};
		\draw (8.25,3)node{$V(G)\setminus D$};
		\draw (4,2.85)node[block1,label=above:$w_1$]{}--(y1);
		\draw (5,2.85)node[block1,label=above:$w_2$]{}--(y2);
		\draw[thick] (5.75,2.85)node{$\cdots$};
		\draw (6.5,2.85)node[block1,label=above:$w_{r-1}$]{}--(yr-1);
	\end{tikzpicture}
	\captionsetup{format=hang}
	\caption{A connected ($P_4\cup rK_1$)-free graph $G$ with $D=\{x_1,x_2,x_3,x_4,y_1,$ $y_2,\ldots,y_{r-1}\}$ $\subseteq V(G)$ such that $G[D]\cong (P_4\cup (r-1)K_1)$. It is possible that $w_i=w_j$ for $i\neq j$.}
	\label{fig:draw-p4-r-1k1}
\end{figure}
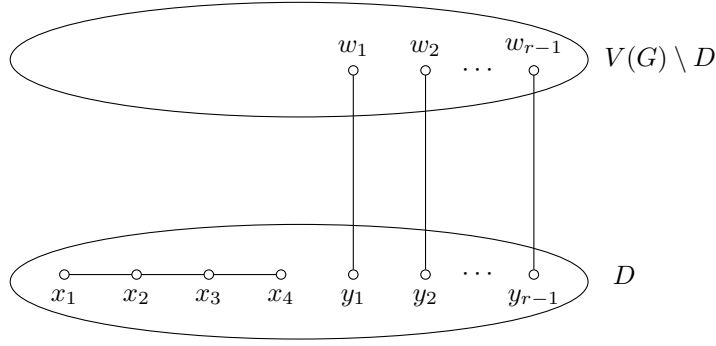
\begin{remark}
	 For $r=1$, the bound given in Lemma~\ref{p4rk1-td-bound} is tight, and a cycle on six vertices ($C_6$) is a connected $(P_4\cup K_1)$-free graph that satisfies this bound. Also, the graphs obtained by replacing every vertex of $C_6$ with a complete graph will remain as a ($P_4\cup K_1$)-free graph with total domination number four. For every $r>1$, there exists a $(P_4\cup rK_1)$-free graph $G_r$ with $\gamma_t(G_r)=2r+1$ (see Remark~\ref{rem-tight-p4-rk1} and Fig.~\ref{fig:tight-p4rk1-op}). Therefore, the bound given in Lemma~\ref{p4rk1-td-bound} is a near-optimal upper bound.
	\label{rem-td-p4-k1-tight}
\end{remark}

\noindent The following corollary follows from Lemmas~\ref{p4-td-bound} and \ref{p4rk1-td-bound}, and the fact that for every $r\geq 1$, every $(K_2\cup rK_1)$-free graph is a ($P_4\cup (r-1)K_1$)-free graph.
\begin{corollary}
	For every $r\geq 1$, the total domination number of a connected $(K_2\cup rK_1)$-free graph is bounded above by $2r$.
	\label{k2-td-cor} 
\end{corollary}\noindent The following observation that follows from the definition of open packing is helpful in proving Lemma~\ref{p4rk1-op-bound}.
\begin{observation}[\cite{Raja2020}]
	Let $S$ be an open packing in a graph $G$. Then, $G[S]$ is $\{P_3,K_3\}$-free, i.e., $G[S]$ is a union of $K_2$'s and $K_1$'s.
	\label{obs-induced-s}
\end{observation}
\noindent The following lemma proves the bound on the open packing number in the class of $(P_4\cup rK_1)$-free graphs. 
\begin{lemma}
	For $r\geq 1$, if $G$ is a connected ($P_4\cup rK_1$)-free graph, then $\rho^o(G)\leq 2r+1$. \label{p4rk1-op-bound}
\end{lemma}
\begin{proof}
	We use induction on $r\geq 1$ to prove this lemma as well. Firstly, we prove the bound for the case $r=1$. Let $G$ be a connected $(P_4\cup K_1)$-free graph. Then, $\rho^o(G)\leq 2< 2r+1=3$ if $G$ is $P_4$-free by Lemma~\ref{p4-td-bound}. So, assume that $G$ contains a $P_4$. Let $D=\{x_1,x_2,x_3,x_4\}$ be a vertex subset of $G$ that induces a $P_4$ in $G$ with $x_ix_{i+1}\in E(G)$ for $i=1,2,3$. Then, since $G$ is $(P_4\cup K_1)$-free, every vertex in $V(G)\setminus D$ is adjacent to some vertex in $\{x_1,x_2,x_3,x_4\}$. Also, every vertex in $\{x_1,x_2,x_3,x_4\}$ is adjacent to a vertex in $\{x_2,x_3\}$. So, $V(G)=\bigcup\limits_{i=1}^4N(x_i)$. Thus, for any open packing $S$ in $G$, $S=S\cap V(G)=\bigcup\limits_{i=1}^4(S\cap N(x_i))$. Since $S$ is an open packing, $|S\cap N(x_i)|\leq 1$ for $i=1,2,3,4$. Thus, $|S|\leq \sum\limits_{i=1}^4|S\cap N(x_i)|\leq 4$. So, to prove the claim, it is enough to disprove the possibility of an open packing of size four in $G$. On the contrary, assume that there exists an open packing $S$ of size four in $G$. Then, by the above arguments, we can conclude that $|S\cap N(x_i)|=1$ for every $i=1,2,3,4$. Let $S\cap N(x_i)=\{z_i\}$. So, $S=\{z_1,z_2,z_3,z_4\}$. This, together with the assumption that $|S|=4$, implies $z_i\neq z_j$ (and so $x_iz_j\notin E(G)$) for $i,j\in \{1,2,3,4\}$ and $i\neq j$. Note that $z_1x_1\in E(G)$ (i.e., $z_1\neq x_1$) whereas $x_1x_3,x_1x_4\notin E(G)$, so $z_1\notin \{x_1,x_3,x_4\}$. Further, $z_1\neq x_2$ since $x_2x_3\in E(G)$ and $z_1x_3\notin E(G)$. Similarly, it can be proved that $z_4\notin \{x_1,x_2,x_3,x_4\}$. Therefore, $z_1z_4\in E(G)$, else $\{z_1,x_1,x_2,x_3\}\cup \{z_4\}$ induces a $(P_4\cup K_1)$ in $G$, a contradiction. Since $z_1z_4\in E(G)$, we have $z_1z_2,z_2z_4\notin E(G)$ because of Observation~\ref{obs-induced-s}. This, together with the fact that $z_1x_1,z_4x_4\in E(G)$, implies that $z_2\notin \{ x_1,x_4\}$. Further, since $x_1z_2,z_2x_4\notin E(G)$, the vertex subset $\{x_1,z_1,z_4,x_4\}\cup \{z_2\}$ induces a $(P_4\cup K_1)$ in $G$, a contradiction. So $\rho^o(G)\leq 3$. Next, we prove the case $r\geq 2$.\\
	\textit{Induction Assumption}: Assume that for every connected $(P_4\cup (r-1)K_1)$-free graph $G'$, $\rho^o(G')\leq 2(r-1)+1$.\\
	Let $G$ be a $(P_4\cup rK_1)$-free graph. If $G$ does not contain any vertex subset that induces a $(P_4\cup (r-1)K_1)$, then $G$ is $(P_4\cup (r-1)K_1)$-free, and so by induction assumption, $\rho^o(G)\leq 2(r-1)+1< 2r+1$. So, assume that $G$ has a vertex subset $D=\{x_1,x_2,x_3,x_4,y_1,y_2,\ldots,y_{r-1}\}$ that induces a $(P_4\cup (r-1)K_1)$ in $G$ with $x_ix_{i+1}\in E(G)$ for $i=1,2,3$. Since $G$ is $(P_4\cup rK_1)$-free, every vertex in $V(G)\setminus D$ is adjacent to some vertex in $D$, i.e., $D$ is a dominating set in $G$. So, $V(G)=N[x_1]\cup N[x_2]\cup N[x_3]\cup N[x_4]\cup N[y_1]\cup N[y_2]\cup \ldots \cup N[y_{r-1}]$. Hence, for any open packing $S$ in $G$, $S=\bigcup\limits_{u\in D}(S\cap N[u])$. Note that $|S\cap N[u]|\leq |S\cap N(u)|+|S\cap \{u\}|\leq 2$. Also, since $x_1,x_2,x_3,x_4\in N(x_2)\cup N(x_3)$, $\bigcup\limits_{i=1}^4N[x_i]=\bigcup\limits_{i=1}^4N(x_i)$. So, $|S\cap (\bigcup\limits_{i=1}^4 N[x_i])|= |S\cap (\bigcup\limits_{i=1}^4 N(x_i))|=|\bigcup\limits_{i=1}^4(S\cap N(x_i))|\leq 4$. Therefore, $|S|=|S\cap V(G)|=|S\cap (N[x_1]\cup N[x_2]\cup N[x_3]\cup N[x_4]\cup N[y_1]\cup N[y_2]\cup \ldots \cup N[y_{r-1}])|=|\bigcup\limits_{i=1}^4(S\cap N(x_i))|+|\bigcup\limits_{j=1}^{r-1}(S\cap N[y_j])|\leq 4+2(r-1)=2r+2$. Thus, to prove the claim, it is enough to rule out the possibility of $|S|=2r+2$. On the contrary, assume that there exists an open packing $S$ of size $2r+2$ in $G$. Then, by the above arguments, we can conclude that $|S\cap N[y_j]|=2$ for $j=1,2,\ldots,r-1$ and $|S\cap (\bigcup\limits_{i=1}^4 N(x_i))|= 4$. Since $|S\cap N(y_i)|\leq 1$ and $|S\cap N[y_i]|=2$, $y_i\in S$ and $S\cap N(y_i)\neq \emptyset$. Let $S\cap N[y_j]=\{w_j,y_j\}$ for $j=1,2,\ldots,r-1$, and let $S\cap (\bigcup\limits_{i=1}^4 N(x_i))=\{z_1,z_2,z_3,z_4\}$ with $x_iz_i\in E(G)$ (then, since $|S\cap N(x_i)|\leq 1$, we have $S\cap N(x_i)=\{z_i\}$ for every $i=1,2,3,4$). So, $S=\left(\bigcup\limits_{i=1}^4\{z_i\}\right)\cup \left(\bigcup\limits_{j=1}^{r-1}\{y_j\}\right)\cup \left(\bigcup\limits_{j=1}^{r-1}\{w_j\}\right)$. This, together with the fact that $|S|=2r+2$, imply that every vertex in the $(2r+2)$-tuple $(z_1,z_2,z_3,z_4,y_1,y_2,\ldots,y_{r-1},w_1,w_2,\ldots,w_{r-1})$ is distinct. This implies that for $s\in \{1,2,3,4\}$, $z_su\in E(G)$ for some $u\in D$ if and only if $u=x_s$ because for every $u\in D$, (i) $S\cap N(u)=\{z_i\}$ if $u=x_i$ for $i\in\{1,2,3,4\}$ and (ii) $S\cap N(u)=\{w_j\}$ if $u=y_j$ for $j\in \{1,2,\ldots,r-1\}$. We complete our claim in two cases based on whether $z_1z_4\in E(G)$ or not.\\
	\emph{Case 1:} $z_1z_4\in E(G)$.\\
	Since $G[S]$ is $\{P_3,K_3\}$-free by Observation~\ref{obs-induced-s}, $z_1z_2,z_2z_4\notin E(G)$. So, $z_2\neq x_1$ since $z_1x_1\in E(G)$. Then, $\{x_1,z_1,z_4,x_4\}\cup \{z_2\}\cup\{y_1\}\cup\{y_2\}\cup \cdots\cup \{y_{r-1}\}$ induces a $(P_4\cup rK_1)$ in $G$, a contradiction.\\
	\emph{Case 2:} $z_1z_4\notin E(G)$.\\
	Then, $\{z_1,x_1,x_2,x_3\}\cup \{z_4\}\cup\{y_1\}\cup\{y_2\}\cup \cdots\cup \{y_{r-1}\}$ induces a $(P_4\cup rK_1)$ in $G$, a contradiction.\\
	As we have arrived at a contradiction in both cases, $\rho^o(G)\leq 2r+1$.
\end{proof}

\begin{remark}
	The bound given in Lemma~\ref{p4rk1-op-bound} is tight. For example, let $G_r$ be a graph obtained by subdividing exactly $r$ edges of $K_{1,r+1}$ by four vertex paths. Graph $G_3$ is shown in Fig.~\ref{fig:tight-p4rk1-op}. The vertex and edge set of $G_r$ can be defined as $V(G_r)=(\cup_{i=1}^r\{x_i,y_i,z_i\})\cup \{u,v\}$ and $E(G_r)=(\cup_{i=1}^r\{x_iy_i,y_iz_i,z_iu\})\cup \{uv\}$. Then, $G_r$ is ($P_4\cup rK_1$)-free (proof of this statement is given as Claim~\ref{claim-gr-of-pr-rk1} in Appendix~\ref{app-sec:claim_Gr}), and $S_r=(\cup_{i=1}^r\{x_i,y_i\})\cup \{v\}$ is an open packing in $G$ of size $2r+1$. Also, note that the graph obtained by replacing the vertex $u$ in $G_r$ with a complete graph will remain a $(P_4\cup rK_1)$-free graph with open packing number $2r+1$. Further, $\gamma_t(G_r)=\rho^o(G)=2r+1$ with $D=\{u,z_1,z_2,\ldots,z_r,y_1,y_2,\ldots,y_r\}$ as a total dominating set in $G_r$.
	\label{rem-tight-p4-rk1}
\end{remark}
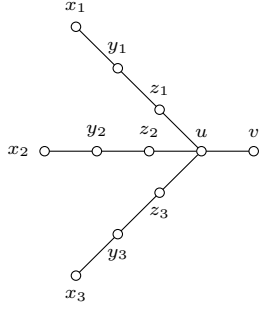
\begin{figure}
	\centering
	\begin{tikzpicture}[scale=0.55]
		\draw (0,0)node[block1,label=above:\scriptsize $x_1$]{}--(1,-1)node[block1,label=above:\scriptsize$y_1$]{}--(2,-2)node[block1,label=above:\scriptsize$z_1$]{}--(3,-3)node(u)[block1,label=above:\scriptsize$u$]{}--(4.25,-3)node[block1,label=above:\scriptsize$v$]{};
		\draw (-.75,-3)node[block1,label=left:\scriptsize$x_2$]{}--(0.5,-3)node[block1,label=above:\scriptsize$y_2$]{}--(1.75,-3)node[block1,label=above:\scriptsize$z_2$]{}--(u);
		\draw (0,-6)node[block1,label=below:\scriptsize$x_3$]{}--(1,-5)node[block1,label=below:\scriptsize$y_3$]{}--(2,-4)node[block1,label=below:\scriptsize$z_3$]{}--(u);
	\end{tikzpicture}
	\captionsetup{format=hang}
	\caption{A ($P_4\cup 3K_1$)-free graph $G_3$ defined in Remark~\ref{rem-tight-p4-rk1} with an open packing $S_3=\{x_1,x_2,x_3,y_1,y_2,y_3,v\}$ of size $7=(2(3)+1)$. Also, $\gamma_t(G_3)=7$ with $D=\{u,z_1,y_1,z_2,y_2,z_3,y_3\}$ as a total dominating set in $G_3$.}
	\label{fig:tight-p4rk1-op}
\end{figure}
\subsubsection*{\boldmath ($P_3\cup rK_1$)-free Graphs}
In this section, we prove that for $r\geq 1$,  $\gamma_t(G)\leq 2r+1$, and $\rho^o(G)\leq 2r$ for every connected ($P_3\cup rK_1$)-free graph $G$. Further, we show that the bound on (i) the open packing number in this graph class is tight and (ii) the total domination number is tight for the class of $(P_3\cup K_1)$-free graphs. 
\begin{remark}
	Every connected $P_3$-free graph is a complete graph, and hence $\rho^o(G)\leq \gamma_t(G)= 2$.
	\label{rem-p3-free}
\end{remark}
\begin{lemma}
	For $r\geq 1$, if $G$ is a connected ($P_3\cup rK_1$)-free graph, then $\gamma_t(G)\leq 2r+1$.	\label{p3rk1-td-bound}
\end{lemma}
\begin{proof}
	Firstly, we prove the case $r=1$. Let $G$ be a connected ($P_3\cup K_1$)-free graph. If $G$ is $P_3$-free, then $\gamma_t(G)= 2$. So, assume that $G$ contains a $P_3$. Let $D=\{x_1,x_2,x_3\}$ be a subset of $V(G)$ that induces a $P_3$ in $G$ with $x_1x_2,x_2x_3\in E(G)$. Since $G$ is ($P_3\cup K_1$)-free, every vertex in $V(G)\setminus D$ is adjacent to some vertex in $D$. Further, since the vertices in $\{x_1,x_2,x_3\}$ are adjacent to a vertex in $\{x_1,x_2\}\subseteq D$, every vertex in $V(G)=(V(G)\setminus D)\cup D$ is adjacent to some vertex in $D$. Hence, $D$ is a total dominating set in $G$, and $\gamma_t(G)\leq |D|=3=2r+1$. Next, we prove the case $r\geq 2$ by induction on $r$.\\
	\textit{Induction Assumption}: For $r\geq 2$, if $G'$ is a connected ($P_3\cup (r-1)K_1$)-free graph, then $\gamma_t(G')\leq 2(r-1)+1$.\\
	Let $G$ be a ($P_3\cup rK_1$)-free graph. Suppose that $G$ does not contain ($P_3\cup (r-1)K_1$) as an induced subgraph, then $G$ is ($P_3\cup (r-1)K_1$)-free, and so by induction assumption, $\gamma_t(G)\leq 2(r-1)+1<2r+1$. So, assume that $G$ contains $(P_3\cup (r-1)K_1)$ as an induced subgraph. Let $D=\{x_1,x_2,x_3,y_1,y_2,\ldots,y_{r-1}\}$ be a subset of $V(G)$ that induces $(P_3\cup (r-1)K_1)$ in $G$ with $x_1x_2,x_2x_3\in E(G)$. Since $G$ is ($P_3\cup rK_1$)-free, every vertex in $V(G)\setminus D$ is adjacent to a vertex in $D$. Also, since $D$ is connected, for every isolated vertex $y_j$ of $G[D]$, there exists a vertex $w_j\in V(G)$ such that $w_jy_j\in E(G)$. Further, every vertex in  $\{x_1,x_2,x_3\}$ is adjacent to a vertex in $\{x_1,x_2\}\subseteq D$. Thus, every vertex in $V(G)=(V(G)\setminus D)\cup D$ is adjacent to some vertex in $D\cup \{w_j\,:\,1\leq j\leq r-1\}$. Hence, $D\cup \{w_j\,:\,1\leq j\leq r-1\}$ is a total dominating set in $G$. So, $\gamma_t(G)\leq |D\cup \{w_j\,:\,1\leq j\leq r-1\}|=3+(r-1)+(r-1)=2r+1$. 
\end{proof}
\noindent Next, we prove the bound on the open packing number in the class of $(P_3\cup rK_1)$-free graphs.
\begin{lemma}
	For $r\geq 1$, if $G$ is a connected ($P_3\cup rK_1$)-free graph, then $\rho^o(G)\leq 2r$.
	\label{p3rk1-op-bound}
\end{lemma}
\noindent Proof of Lemma~\ref{p3rk1-op-bound} is similar to the proof of Lemma~\ref{p4rk1-op-bound}. However, the complete proof is given in Appendix~\ref{app-sec:claim_Gr}. The corollary below follows from Lemmas~\ref{p3rk1-td-bound} and \ref{p3rk1-op-bound}, and the fact that for $r\geq 3$, every $rK_1$-free graph is a ($P_3\cup (r-2)K_1$)-free graph as well.
\begin{corollary}
	For $r\geq 3$, let $G$ be a connected $rK_1$-free graph. Then, (i) $\gamma_t(G)\leq 2r-3$ and (ii) $\rho^o(G)\leq 2(r-2)$.
	\label{op-td-bound-cor}
\end{corollary}
\begin{remark}
	For $r\geq 3$, the bound on the open packing number of $rK_1$-free graphs given in Corollary~\ref{op-td-bound-cor} (and so the one in Lemma~\ref{p3rk1-op-bound}) is tight. For example, let $H_r$ be a graph obtained by identifying a leaf vertex of a (distinct) $P_3$ with every vertex of a complete graph $K_{r-2}$, then $H_r$ is a $rK_1$-free graph (and so ($P_3\cup (r-2)K_1$)-free graph). An example of $H_r$ for $r=7$ is given in Fig.~\ref{fig:tight-op-rk1}. The vertex and edge set of $H_r$ can be defined as $V(H_r)=\cup_{i=1}^{r-2}\{x_i,y_i,z_i\}$ and $E(H_r)=E_{r1}\cup E_{r2}$, where $E_{r1}=\cup_{i=1}^{r-2}\{x_iy_i,y_iz_i\}$ and $E_{r2}=\{z_iz_j\,:\, 1\leq i<j\leq r-2\}$.  Then, $S_r=\{x_1,x_2,\ldots,x_{r-1},$ $y_1,y_2,\ldots,y_{r-2}\}$ is an open packing (as well as a total dominating set) of size $2(r-2)$ in $H_r$. This also implies that the bound given in Lemma~\ref{p3rk1-td-bound} is a near-optimal upper bound. Also, for the case $r=1$, the bound given in Lemma~\ref{p3rk1-td-bound} is tight, and a cycle on five vertices ($C_5$) is a connected $(P_3\cup K_1)$-free graph that satisfies this bound. Further, the graphs obtained by replacing every vertex of $C_5$ with a complete graph will remain as a ($P_3\cup K_1$)-free graph satisfying this bound.
	\label{rem-rk1-op}
\end{remark}
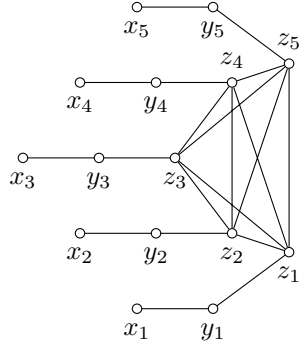
\begin{figure}
	\centering
	\begin{tikzpicture}
		\draw (1.5,0)node[block1,label=below:$x_1$]{}--(2.5,0)node[block1,label=below:$y_1$]{}--(3.5,0.75)node(z1)[block1,label=below:$z_1$]{};
		\draw (0.75,1)node[block1,label=below:$x_2$]{}--(1.75,1)node[block1,label=below:$y_2$]{}--(2.75,1)node(z2)[block1,label=below:$z_2$]{};
		\draw (0,2)node[block1,label=below:$x_3$]{}--(1,2)node[block1,label=below:$y_3$]{}--(2,2)node(z3)[block1,label=below:$z_3$]{};
		\draw (0.75,3)node[block1,label=below:$x_4$]{}--(1.75,3)node[block1,label=below:$y_4$]{}--(2.75,3)node(z4)[block1,label=above:$z_4$]{};
		\draw (1.5,4)node[block1,label=below:$x_5$]{}--(2.5,4)node[block1,label=below:$y_5$]{}--(3.5,3.25)node(z5)[block1,label=above:$z_5$]{};
		\draw (z1)--(z2)--(z3)--(z4)--(z5)--(z1)--(z3)--(z5)--(z2)--(z4)--(z1);
	\end{tikzpicture}
	\captionsetup{format=hang}
	\caption{A connected $7K_1$-free graph $H_7$ with an open packing (as well as a total dominating set) $S_7=\{x_1,x_2,x_3,$ $x_4,x_5,y_1,y_2,y_3,y_4,y_5\}$ of size $2(7-2)=10$. Note that $H_7$ is also a ($K_2\cup 6K_1$)-free graph.}
	\label{fig:tight-op-rk1}
\end{figure}
\subsubsection*{\boldmath $(K_2\cup rK_1)$-free Graphs}
For $r\geq 1$, it is known by Corollary~\ref{k2-td-cor} that $\gamma_t(G)\leq 2r$ for every connected ($K_2\cup rK_1$)-free graph. In this section, we show that the open packing number of a connected ($K_2\cup rK_1$)-free graph is at most $\max\{r+1,2(r-1)\}$. 
\begin{remark}
	Let $G$ be a connected ($K_2\cup K_1$)-free graph. Then, $G$ is also a connected $P_4$-free graph. So, $\rho^o(G)\leq \gamma_t(G)=2(=r+1)$. 
	\label{rem-k2k1-op-bound}
\end{remark}
\begin{lemma}
	For $r\geq 2$, the open packing number of a connected ($K_2\cup rK_1$)-free graph is at most $\max\{r+1,2(r-1)\}$.
	\label{k2-rk1-op}
\end{lemma}
\begin{proof}
	We prove the lemma in two cases. \\
	\emph{Case 1:} $r=2$.\\
	Then, $\max\{r+1,2(r-1)\}=r+1=3$. Since every connected ($K_2\cup 2K_1$)-free graph $G$ is also a connected ($P_4\cup K_1$)-free graph, $\rho^o(G)\leq 2(r-1)+1=3$ by Lemma~\ref{p4rk1-op-bound}. \\
	\emph{Case 2:} $r\geq 3$.\\
	Then, $\max\{r+1,2(r-1)\}=2(r-1)$. We prove that $\rho^o(G)\leq 2(r-1)$. On the contrary, assume that there exists a  connected ($K_2\cup rK_1$)-free graph $G$ with $\rho^o(G)> 2(r-1)$. We know that every open packing of size at least $2(r-1)+1=2r-1$ contains an open packing of size $2r-1$. So, to contradict the assumption, it is enough to eliminate the possibility of an open packing of size $2r-1$ in $G$. Let $S$ be an open packing of size $2r-1$ in $G$. Then, since $G[S]$ is $\{P_3,K_3\}$-free, every component of $G[S]$ contains at most two vertices, and so the number of components in $G[S]$ is at least $\lceil\frac{2r-1}{2}\rceil>r-1$. Also, since $2r-1$ is odd, there exists at least one isolated vertex in $G[S]$ due to the fact that every component of $G[S]$ contains at most two vertices. Let $u$ be the isolated vertex in $G[S]$. Since $G$ is connected, there exists a vertex $v\in V(G)$ such that $uv\in E(G)$. Note that $G[S\setminus\{u\}]$ contains at least $(r-1)$ components. Let $G_1,G_2,\ldots,G_p$ be the components of $G[S\setminus\{u\}]$ for some $p\geq r-1$. We now contradict our assumption in two cases based on the value of $p$.\\
	\emph{Case 2.1:} $p\geq r$.\\
	Let $v_i\in V(G_i)$ for $i=1,2,\ldots,r$. Since each $v_i$ is from a distinct component of $G[S\setminus\{u\}]$, $v_i\neq v_j$ for $i,j\in \{1,2,\ldots,r\}$ and $i\neq j$. Also, since $v_i\in S\setminus \{u\}$, $v_i\neq u$ and $v_iu\notin E(G[S])\subseteq E(G)$ by the fact that $u$ is an isolated vertex in $G[S]$. This implies that $v_i\neq v$ for $i=1,2,\ldots,r$ since $uv\in E(G)$. Further $v_iv\notin E(G)$ since $v\in N(u)$ and $N(u)\cap N(v_i)$ is empty because $S$ is an open packing in $G$ with $u,v_i\in S$ for every $i=1,2,\ldots,r$. This implies that $\{u,v\}\cup \{v_1,v_2,\ldots,v_r\}$ induces a ($K_2\cup rK_1$) in $G$, a contradiction.\\
	\emph{Case 2.2:} $p=r-1$.\\
	Then, $G_1,G_2,\ldots,G_{r-1}$ are the components in $G[S\setminus \{u\}]$. Recall that $|V(G_i)|\leq 2$ for every $1\leq i\leq r-1$. But since $2r-2=|S\setminus\{u\}|=|\bigcup\limits_{i=1}^{r-1}V(G_i)|\leq 2r-2$, every $G_i$ consists of exactly two vertices. Let $V(G_i)=\{u_i,v_i\}$ for $1\leq i\leq r-1$. Since $G$ is connected, there exists a shortest path joining $u_1$ and $u$. Let $x_1x_2\ldots x_s$ with $u=x_1$ and $u_1=x_s$ be a shortest path joining $u_1$ and $u$ in $G$ (see Fig~\ref{addon-fig}). Since $S$ is an open packing in $G$ with $u,v_1\in S$, $N(u)\cap N(v_1)$ is empty. So, $uu_1\notin E(G)$ since $u_1v_1\in E(G)$. Similarly, $N(u)\cap N(u_1)$ is empty. This implies that $s\geq 4$. So, $u\notin \{x_{s-2},x_{s-1}\}$. We further divide this case into two cases based on whether $x_{s-1}=v_1$ or not.\\
	\emph{Case 2.2.1:} $x_{s-1}\neq v_1$.\\
	Then, since $x_{s-1}x_s=x_{s-1}u_1\in E(G)$, i.e., $x_{s-1}\in N(u_1)$ and $S$ is an open packing with $u_1,v_i\in S$, $x_{s-1}\notin N(v_i)$ (i.e., $x_{s-1}v_i\notin E(G)$) for every $i=1,2,\ldots,r-1$. For similar reasons, $x_{s-1}u,x_{s-1}u_2\notin E(G)$. The existence of $u_2$ is guaranteed by the fact that $r-1\geq 2$. Also, note that $uv_i,u_2v_\ell,v_iv_j\notin E(G[S])\subseteq E(G)$ for every $i,j,\ell\in \{1,2,\ldots,r-1\}$ and $i\neq j$, $\ell \neq 2$ (see Fig.~\ref{addon-fig}). Thus, $\{u_2,v_2\}\cup \{x_{s-1}\}\cup\{v_1\}\cup\{v_3\}\cup\{v_4\}\cup\ldots\cup\{v_{r-1}\}\cup\{u\}$ is a ($K_2\cup rK_1$) in $G$, a contradiction.\\
	\emph{Case 2.2.2:} $x_{s-1}=v_1$.\\
	Then, similar to the above case, it can be proved that $x_{s-2}u,x_{s-2}u_i,x_{s-2}v_2\notin E(G)$ for every $i=1,2,\ldots,r-1$. Also, $uu_i,v_2u_\ell,u_iu_j\notin E(G)$ for every $i,j,\ell\in \{1,2,\ldots,r-1\}$ and $i\neq j$, $\ell \neq 2$. So, $\{u_2,v_2\}\cup \{x_{s-2}\}\cup\{u_1\}\cup\{u_3\}\cup\{u_4\}\cup\ldots\cup\{u_{r-1}\}\cup\{u\}$ is a ($K_2\cup rK_1$) in $G$, a contradiction.
\end{proof}
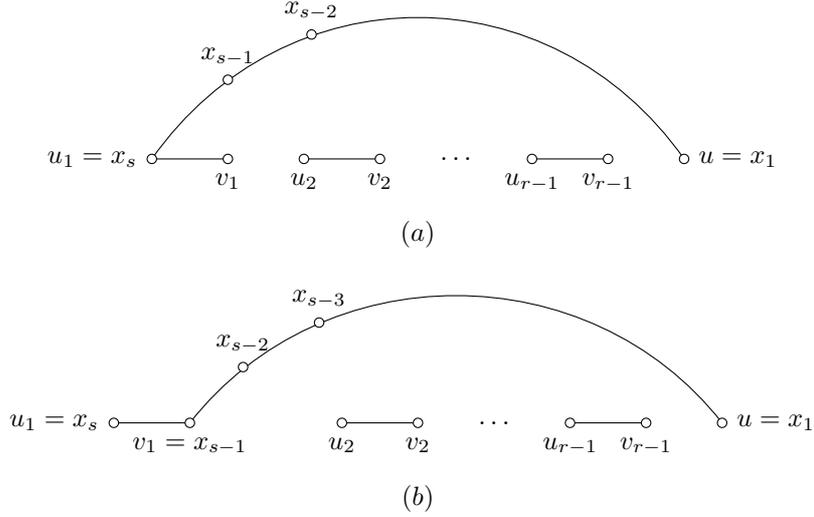
\begin{figure}
	\centering
	\begin{tikzpicture}
		\draw (7,0)node(u)[block1,label=right:${u=x_1}$]{}..controls(5.25,2.5)and(1.75,2.5)..(0,0)node(u1)[block1,label=left:${u_1=x_s}$]{}--(1,0)node(v1)[block1,label=below:$v_1$]{};
		\draw (2,0)node(u2)[block1,label=below:${u_2}$]{}--(3,0)node(v2)[block1,label=below:$v_2$]{};
		\draw (5,0)node(ur)[block1,label=below:${u_{r-1}}$]{}--(6,0)node(vr)[block1,label=below:$v_{r-1}$]{};
		\draw (4,0)node{$\ldots$};
		\draw (1,1.05)node(xs)[block1,label=above:$x_{s-1}$]{};
		\draw (2.1,1.65)node(xs2)[block1,label=above:$x_{s-2}$]{};
		\draw (3.5,-1)node{$(a)$};
		\draw (7.5,-3.5)node(u)[block1,label=right:${u=x_1}$]{}..controls(5.75,-1.25)and(2.25,-1.25)..(0.5,-3.5)node(v1)[block1,label=below:${v_1=x_{s-1}}$]{}--(-0.5,-3.5)node(v1)[block1,label=left:${u_1=x_s}$]{};
		\draw (2.5,-3.5)node(u2)[block1,label=below:${u_2}$]{}--(3.5,-3.5)node(v2)[block1,label=below:$v_2$]{};
		\draw (5.5,-3.5)node(ur)[block1,label=below:${u_{r-1}}$]{}--(6.5,-3.5)node(vr)[block1,label=below:$v_{r-1}$]{};
		\draw (1.2,-2.76)node[block1,label=above:$x_{s-2}$]{};
		\draw (2.2,-2.17)node[block1,label=above:$x_{s-3}$]{};
		\draw (4.5,-3.5)node{$\ldots$};
		\draw (3.5,-4.5)node{$(b)$};
	\end{tikzpicture}
	\captionsetup{format=hang}
	\caption{(a) $G[S]$ along with the path $x_1x_2\ldots x_s$ joining $u=x_1$ and $u_1=x_s$ in $G$ such that $x_{s-1}\neq v_1$ and (b) $G[S]$ along with the path $x_1x_2\ldots x_s$ joining $u=x_1$ and $u_1=x_s$ in $G$ such that $x_{s-1}= v_1$. Note that $\{x_2,x_3,\ldots,x_{s-2}\}\cap \{u_2,v_2,\ldots,u_{r-1},v_{r-1}\}$ need not be empty.}
	\label{addon-fig}
\end{figure}
\begin{remark}
	The bound given in Lemma~\ref{k2-rk1-op} is tight. For $r\geq 3$, let $H_{r+1}$ be the graph defined in Remark~\ref{rem-rk1-op}. Then, $H_{r+1}$ is a ($K_2\cup rK_1$)-free graph, and the set $S_{r+1}$ in Remark~\ref{rem-rk1-op} is an open packing of size $2(r-1)$. 
\end{remark}
\noindent Using the bounds obtained in the subclasses of $(P_4\cup rK_1)$-free graphs studied so far, we show that \tdset\ and \openpack\ are in P in the subclasses of $(P_4\cup rK_1)$-free graphs.
\begin{theorem}
	Given a graph class $\mathcal{G}$, \tdset\ is in P on the class $\mathcal{G}$ if $\mathcal{G}$ is (a) ($P_4\cup rK_1$)-free graphs, (b) $(P_3\cup rK_1)$-free graphs, (c) $(K_2\cup rK_1)$-free graphs, or (d) $rK_1$-free graphs.
	\label{thm-p4rk1-td-ptime}
\end{theorem}
\noindent The above theorem follows from Lemmas~\ref{lem-td-opt-ell}, \ref{p4-td-bound}, \ref{p4rk1-td-bound}, and \ref{p3rk1-td-bound}, as well as Corollaries~\ref{k2-td-cor} and \ref{op-td-bound-cor}. The graph classes $(P_3\cup rK_1)$-free graphs, $(K_2\cup rK_1)$-free graphs, and $rK_1$-free graphs are explicitly stated in Theorem~\ref{thm-p4rk1-td-ptime} for clarity. 
\begin{theorem}
	Given a graph class $\mathcal{G}$, if $\mathcal{G}$ is (a) $(P_4\cup rK_1)$-free graphs ($r\geq 0$), (b) $(P_3\cup rK_1)$-free graphs ($r\geq 0$), (c) $(K_2\cup rK_1)$-free graphs ($r\geq 1$), or (d) $rK_1$-free graphs ($r\geq 3$), then for every fixed $r$ and for every graph $G\in \mathcal{G}$,  (i) there are polynomially many open packings in $G$, (ii) all open packings in $G$ can be found in polynomial time, and hence $\rho^o(G)$ can be found in polynomial time.
	\label{thm-p4-ptime}
\end{theorem}
\noindent The above theorem follows from Lemmas~\ref{lem-tc-bound}, \ref{p4-td-bound}, \ref{p4rk1-op-bound}, \ref{p3rk1-op-bound}, \ref{k2-rk1-op}, Corollary~\ref{op-td-bound-cor}, Remarks~\ref{rem-p3-free}, \ref{rem-k2k1-op-bound}, and the fact that for a graph $G$ with components $G_1,G_2,\ldots,G_k$, $\rho^o(G)=\sum_{i=1}^{k}\rho^o(G_i)$. Next, we prove Theorem~\ref{thm-h-dic}.
\begin{customthm}{\ref{thm-h-dic}}
	For $p\geq 4$, let $H$ be a graph on $p$ vertices. Then, \openpack\ is polynomial time solvable on the class of $H$-free graphs if and only if $H\in \{pK_1,(K_2\cup (p-2)K_1),(P_3\cup (p-3)K_1),(P_4\cup(p-4)K_1)\}$ unless P = NP.
\end{customthm}
\begin{proof}
	Let $H$ be a graph on $p$ vertices for some $p\geq 4$ such that \openpack\ is polynomial time solvable on $H$-free graphs. Then, we prove that $H\in\{P_4\cup (p-4)K_1,P_3\cup (p-3)K_1, K_2\cup (p-2)K_1, pK_1\}$ unless P = NP. On the contrary, assume that $H\notin\{P_4\cup (p-4)K_1,P_3\cup (p-3)K_1, K_2\cup (p-2)K_1, pK_1\}$. Then, by Observation~\ref{obs-h-free}, $H$ contains at least one of (i) $K_3$, (ii) $2K_2$, (iii) $C_4$, (iv) $K_{1,3}$, or (v) $C_5$ as its induced subgraph. So, (a) $H'$-free graphs are a subclass of $H$-free graphs for some $H'\in \{K_3,2K_2,C_4,K_{1,3},C_5\}$ and (b) \openpack\ is NP-complete on $H'$-free graphs for every $H'\in \{K_3,2K_2,C_4,K_{1,3},C_5\}$ (by Remark~\ref{rem-triangle} and Theorem~\ref{thm-k-1,3-npc}). Thus, by (a), (b), and the fact that \openpack\ is in NP in every subclass of simple graphs, we can conclude that \openpack\ is NP-complete on $H$-free graphs, a contradiction. This proves the necessary part of Theorem~\ref{thm-h-dic}. \\
	The sufficiency part of Theorem~\ref{thm-h-dic} is held by Theorem~\ref{thm-p4-ptime}.
\end{proof}
 \noindent Note that  Observation~\ref{obs-h-free}, along with Theorem~\ref{thm-p4rk1-td-ptime}, and the fact that \tdset\ is NP-complete for (i) Split graphs ($\{2K_2,C_4,C_5\}$-free graphs)~\cite{CORNEIL1984}, (ii) Bipartite graphs (a subclass of $K_3$-free graphs)~\cite{npc-tdset-bigraphs1983}, and (iii) Line graphs (a subclass of $K_{1,3}$-free graphs)~\cite{McRae1995}, imply the following theorem, similar to Theorem~\ref{thm-h-dic}, for \tdset.
\begin{theorem}
	For $p\geq 4$, let $H$ be a graph on $p$ vertices. Then, \tdset\ is polynomial time solvable on the class of $H$-free graphs if and only if $H\in \{pK_1,(K_2\cup (p-2)K_1),(P_3\cup (p-3)K_1),(P_4\cup(p-4)K_1)\}$ unless P = NP.	\label{thm-h-dic-tdset}
\end{theorem}
\subsection{\boldmath $H$ on three vertices}\label{app-sec:h-on-three-vertices}
 In this section, we give the proof of Theorem~\ref{h-on-three-vertices}.
\begin{customthm}{\ref{h-on-three-vertices}}
	Let $H$ be a graph on three vertices. Then, an optimal open packing on $H$-free graphs can be found in polynomial time if and only if $H\ncong K_3$ unless P = NP.
\end{customthm}
\begin{proof}
	Note that $3K_1,(K_2\cup K_1), P_3$, and $K_3$ are the only non-isomorphic graph on three vertices. By Theorem~\ref{thm-p4-ptime}, an optimal open packing of a $H$-free graph can be found in polynomial time if $H\in \{3K_1,(K_2\cup K_1), P_3\}$ and \openpack\ is NP-complete on $K_3$-free graphs by Remark~\ref{rem-triangle}. This completes the proof of Theorem~\ref{h-on-three-vertices}.
\end{proof}
\noindent Note that \tdset\ is NP-complete for $K_3$-free graphs~\cite{npc-tdset-bigraphs1983}. So, by Theorem~\ref{thm-p4rk1-td-ptime}, a result similar to Theorem~\ref{h-on-three-vertices} for \tdset holds. This note completes our study on the complexity (P vs NPC) of \tdset\ and \openpack\ on $H$-free graphs. Next, we study the complexity of these problems on some subclasses of split graphs.
\section{Split Graphs}\label{sec:split}
Ramos et al.~\cite{Igor2014} proved that \openpack\ is NP-complete on split graphs. In this section, we study the complexity of \openpack\ in (i) $K_{1,r}$-free split graphs for $r\geq 2$ and (ii) $I_r$-split graphs for $r\geq 1$. This study shows the complexity (P vs NPC) difference between \tdset\ and \openpack\ in split graphs (see Table~\ref{table-compar-diff}).
\subsection{\boldmath $K_{1,r}$-free Split Graphs}\label{sec:k_1,4-split} It is known that \tdset\ is NP-complete on $K_{1,5}$-free split graphs~\cite{White1995} and is polynomial time solvable on $K_{1,4}$-free split graphs~\cite{RENJITH2020246}. In this section, we give a dichotomy result for \openpack\ in $K_{1,r}$-free split graphs by proving that the problem is (i) NP-complete for $r\geq 4$ and (ii) polynomial time solvable for $r\leq 3$.\\
We begin this section with a construction to prove that \openpack\ is NP-complete on $K_{1,4}$-free split graphs, which is a reduction from \indset, a well-known NP-complete problem~\cite{karp}.
\begin{construct}
	\emph{Input:} A simple graph $G(V,E)$.\\
	\emph{Output:} A $K_{1,4}$-free split graph $G'(C\cup I, E)$.\\
	\emph{Guarantee:} $G$ has an independent set of size $k$ if and only if $G'$ has an open packing of size $k+1$.\\
	\emph{Procedure:}\\[2pt]
	\begin{tabular}{l >{\RaggedRight}p{13.25cm}}
		Step 1:& Subdivide each edge $e=uu'$ in $G$ into a three vertex path $ueu'$ in $G'$.\\[2pt]
		Step 2:& Introduce three new vertices, $x$, $y$, and $z$, and two edges, $xy$ and $xz$, in $G'$.\\[2pt]
		Step 3:& Make $E(G)\cup \{y,z\}$ a clique in $G'$.
	\end{tabular}
	\label{k-1,4-split-op-construct}
\end{construct}
\noindent An example of Construction~\ref{k-1,4-split-op-construct} is shown in Fig.~\ref{k-1,4-split-op-fig}. The vertex set and the edge set of the graph $G'$ are $V(G')=V(G)\cup E(G)\cup \{x,y,z\}$ and $E(G')=\{ue\,:\,  u\in V(G), e\in E(G)\text{ and } e \in E_G(u)\}\cup \{xy,xz,yz\}\cup \{ye\,:\, e\in E(G)\}\cup \{ze\,:\, e\in E(G)\} \cup \{ee'\,:\, e,e'\in E(G)\text{ and }e\neq e'\}$. Note that $|V(G')|=n+m+3$ and $|E(G')|=2m+3+2m+\binom{m}{2}=O(m^2)$. Hence, the graph $G'$ can be constructed in polynomial time from $G$. Also, by Construction~\ref{k-1,4-split-op-construct}, $V(G')=C\cup I$ is a clique-independent set partition of $G'$ with $C=E(G)\cup \{y,z\}$ and $I=V(G)\cup \{x\}$. Thus, $G'(C\cup I, E)$ is a split graph. 
\begin{figure}
	\centering
	\begin{tikzpicture}[scale=.75]
		{
			\draw (0,-1.5)node[block1,label=below:$u_1$]{}--(2,0.5)node(v2)[block1,label=right:$u_2$]{}--(2,3.5)node(u3)[block1,label=right:$u_3$]{};
			\draw (-2,3.5)node(u5)[block1,label=left:$u_4$]{}--(-2,0.5)node(u6)[block1,label=left:$u_5$]{}--(v2);
			\draw (u3)--(u5);
			\draw (u6)--(0.1,2.1)node(u7)[block1,label=below:$u_6$]{}--(u3);
			\draw (-2,2)node(e4)[label=left:$e_4$]{};
			\draw (0,3.5)node(e3)[label=above:$e_3$]{};
			\draw (2,2)node(e2)[label=right:$e_2$]{};
			\draw (1,-.5)node(e1)[label=below:$e_1$]{};
			\draw (0,.5)node(e5)[label=below:$e_5$]{};
			\draw (-1.1,1.5)node(e6)[label=above:$e_6$]{};
			\draw (1.1,2.9)node(e7)[label=below:$e_7$]{};
			\draw (0,-2.9)node{(a) $G$};}
		{\draw (11,1)node(u2)[block1,label=right:$u_2$]{}--(6.5,.5)node(e1)[block1,label=left:$e_1$]{}--(11,0)node(u1)[block1,label=right:$u_1$]{};}
		{\draw (11,2)node(u3)[block1,label=right:$u_3$]{}--(6.5,1.25)node(e2)[block1,label=left:$e_2$]{}--(u2);}
		{	\draw (6.5,3.5)node(e5)[block1,label=left:$e_5$]{}--(11,4)node(u5)[block1,label=right:$u_5$]{}--(6.5,2.75)node(e4)[block1,label=left:$e_4$]{}--(11,3)node(u4)[block1,label=right:$u_4$]{}--(6.5,2)node(e3)[block1,label=left:$e_3$]{}--(u3);
			\draw (e5)--(u2);
			\draw (u5)--(6.5,4.25)node(e6)[block1,label=left:$e_6$]{}--(11,5)node(u6)[block1,label=right:$u_6$]{}--(6.5,5)node(e7)[block1,label=left:$e_7$]{}--(u3);}
		{\draw (6.5,-.25)node(y)[block1,label=left:$y$]{}--(11,-1)node(x)[block1,label=right:$x$]{}--(6.5,-1)node(z)[block1,label=left:$z$]{};}
		{\draw (6.25,2.1)ellipse[x radius=1.55, y radius=3.6];
			\draw (11.25,2.2)ellipse[x radius=1.55, y radius=3.75];
			\draw (6.25,-2.15)node{\small Clique};
			\draw (12,-2.15)node{\small Independent Set};
			\draw (9,-3)node{(b) $G'$};}
	\end{tikzpicture}
	\captionsetup{format=hang}
	\caption{(a) a graph $G$ and (b) the graph $G'$ constructed from $G$ using Construction~\ref{k-1,4-split-op-construct}}
	\label{k-1,4-split-op-fig}
\end{figure}
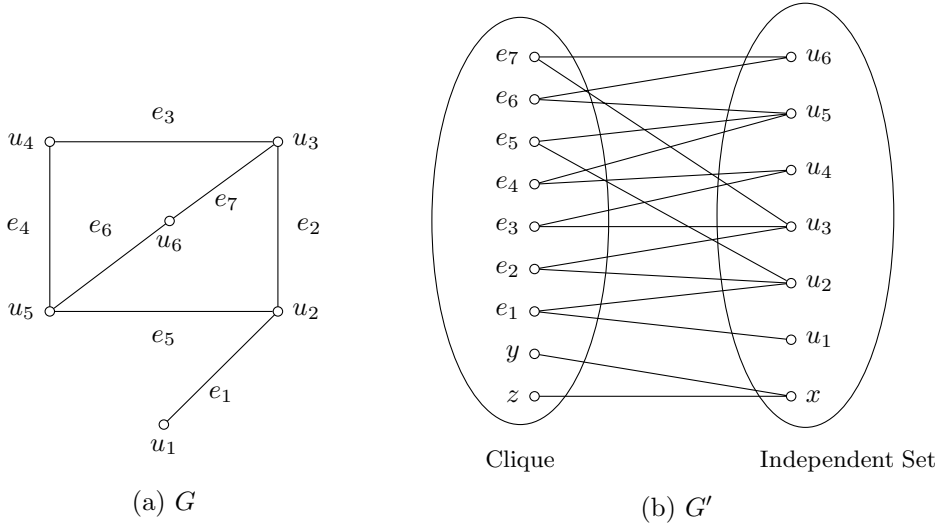
 \begin{claim-alt}
	The graph $G'$ output by Construction~\ref{k-1,4-split-op-construct} is a $K_{1,4}$-free graph.\\
	On the contrary, assume that $\{w,w_1,w_2,w_3,w_4\}$ induces a $K_{1,4}$ in $G'$ with $w$ as its centre. Then, $\{w_1,w_2,w_3,w_4\}$ is an independent set in $G'$. So $w\notin I$, since $N_{G'}(v)\subseteq C$ is a clique for every vertex $v\in I$, whereas $N_{G'}(w)$ contains an independent set on four vertices. This implies that $w\in C$. By Construction~\ref{k-1,4-split-op-construct}, for every vertex $v\in C$, $|N_{G'}(v)\cap I|\leq 2$, and so $w$ has at most two neighbours in $I$. Therefore, $|\{w_1,w_2,w_3,w_4\}\cap C|\geq 2$, a contradiction to $C$ being a clique and $\{w_1,w_2,w_3,w_4\}$ being an independent set in $G'$. Thus, $G'$ is a $K_{1,4}$-free graph.  
	\label{k_1,4-claim}
 \end{claim-alt}
\noindent The following claims help us prove the guarantee of Construction~\ref{k-1,4-split-op-construct}.
 \begin{claim-alt}
	A set $S(\subseteq V(G))$ is an independent set in the input graph $G$ of Construction~\ref{k-1,4-split-op-construct} if and only if $S$ is an open packing in the output graph $G'$.\\
	Proof of Claim~\ref{claim-s-ind-iff-sub} is similar to the proof of Claim~\ref{claim-s-iff-k_1,3}, and so it is moved to Appendix~\ref{app:sec-split}. 
	\label{claim-s-ind-iff-sub}
 \end{claim-alt}
 \begin{claim-alt}
	For every open packing $S$ in the output graph $G'$ of Construction~\ref{k-1,4-split-op-construct}, $|S\cap (C\cup \{x\})|=|S\cap (E(G)\cup \{x,y,z\})|\leq 1$.\\
	Observe that in $G'$, the graph induced by $C\cup \{x\}$ consists of a complete graph ($G'[C]$) on $m+2$ vertices together with a vertex $x$, which is adjacent to exactly two vertices ($y\text{ and }z$) in $C$. So, any two distinct vertices in $C\cup \{x\}$ have a common neighbour. Thus, the claim holds.  
	\label{s-cap-clique-claim}
 \end{claim-alt}
\begin{proof}[Proof of Guarantee of Construction~\ref{k-1,4-split-op-construct}]
	We prove the guarantee of Construction~\ref{k-1,4-split-op-construct} in two cases.\\
	\emph{Case 1}:	 $G$ has an independent set $S$ of size $k$.\\
	Then, by Claim~\ref{claim-s-ind-iff-sub}, $S\subseteq V(G)$ is an open packing in $G'$ of size $k$. To prove that $G'$ has an open packing of size $k+1$, we show that $S\cup \{x\}$ is an open packing in $G'$. Note that by Construction~\ref{k-1,4-split-op-construct}, $N_{G'}(u)\cap N_{G'}(x)\subseteq E(G)\cap \{y,z\}=\emptyset$ for every vertex $u\in V(G)$ (in particular, for every vertex $u\in S$). So, $S\cup \{x\}$ is an open packing of size $k+1$ in $G'$.\\ 
	\emph{Case 2}: $G'$ has an open packing $S$ of size $k+1$.\\
	Then, by Claim~\ref{s-cap-clique-claim}, $|S\cap (E(G)\cup \{x,y,z\})|\leq 1$. So, $S'=S\setminus (E(G)\cup \{x,y,z\})$ is an open packing of size at least $k$ in $G'$. Also, by the construction of $G'$, we have $S'\subseteq V(G)$. Hence, by Claim~\ref{claim-s-ind-iff-sub}, $S'$ is an independent set of size at least $k$ in $G$. \\
	Thus, by Cases 1 and 2, the guarantee of Construction~\ref{k-1,4-split-op-construct} holds.
\end{proof}
\noindent The following theorem is held by Construction~\ref{k-1,4-split-op-construct} and the fact that \indset\ is NP-complete on simple graphs~\cite{karp}.
\begin{theorem}
	\openpack\ is NP-complete on $K_{1,4}$-free split graphs.
	\label{thm-k-1,4-split-npc}
\end{theorem}
\noindent Similar to that of Theorems~\ref{thm-k_1,3-inapprx} and \ref{thm-k_1,3-w1}, it can be proved that on $K_{1,4}$-free split graphs (i) \openpack\ parameterized by solution size is W[1]-complete  and (ii) \maxopenpack\ is hard to approximate within a factor of $N^{(\frac{1}{2}-\epsilon)}$ for any $\epsilon>0$ unless P = NP, where $N$ denotes the number of vertices in a $K_{1,4}$-free split graph (proof of (ii) is given as Theorem~\ref{k-1,4-inapprx} in Appendix~\ref{app:sec-split}).\\
\noindent Next, to complete the dichotomy result of \openpack\ in $K_{1,r}$-free split graphs, we prove that a maximum open packing in $K_{1,3}$-free split graphs can be found in polynomial time, and the following claim and lemma are useful in proving it.
 \begin{claim-alt}
	If $S$ is an open packing in a connected split graph $G(C\cup I, E)$ with $S\cap C\neq \emptyset$, then for every vertex $u\in S\cap C$, $S\subseteq N[u]$, and so $|S|\leq 2$.\\
	 Let $S$ be an open packing in a connected split graph $G(C\cup I, E)$ with $S\cap C\neq \emptyset$. Let $u\in S\cap C$. Then, we show that $S\subseteq N[u]$. On the contrary, assume that there exists a vertex $v\in S\setminus N[u]$. Since $C\subseteq N[u]$ and $v\notin N[u]$, $v\notin C$, i.e., $v\in I$. Since $G$ is connected, $deg(v)\geq 1$, and so there exists $w\in C$ such that $vw\in E(G)$. Since $v\notin N(u)$, $w\neq u$. Further, since $w,u\in C$, $w\in N(u)$. This implies that $w\in N(u)\cap N(v)$, a contradiction to $S$ is an open packing in $G$ and $u,v\in S$. Therefore, $S\subseteq N[u]$. Since $S$ is an open packing in $G$, $|S\cap N(u)|\leq 1$. So, $|S|\leq |S\cap N(u)|+1\leq 2$.	\label{claim-scapc}
 \end{claim-alt}
 \begin{lemma}[\cite{RENJITH2020246}]
 	If $G(C\cup I,E)$ is a $K_{1,3}$ split graph with a vertex $u\in C$ such that $|N(u)\cap I|=2$, then $|I|\leq 3$.\label{renjith-lemma}
 \end{lemma}
\begin{theorem}
	\openpack\ is polynomial time solvable on $K_{1,3}$-free split graphs. \label{thm-k-1,3-split-poly}
\end{theorem}
\begin{proof}
	Let $G(C\cup I,E)$ be a $K_{1,3}$-free split graph. W.L.O.G., assume that $C\neq \emptyset$ and $I\neq \emptyset$. Since $G$ is $K_{1,3}$-free, $|N(u)\cap I|\leq 2$ for every $u\in C$. We show that $\rho^o(G)\in \{1,2,|I|\}$ in two cases based on $|N(u)\cap I|$ for $u \in C$.\\
	\emph{Case 1:} $|N(u)\cap I|\leq 1$ for every $u\in C$.\\
	Let $I=\{x_1,x_2,\ldots,x_k\}$. Let $V_i=N(x_i)\cap C$ for $i=1,2,\ldots,k$ and $V_0=C\setminus(\cup_{i=1}^kV_i)$. Then, $V_i\cap V_j= \emptyset$ for every $i,j\in \{1,2,\ldots,k\}$ and $i\neq j$. If not, let $u'\in V_i\cap V_j$ for some $i,j\in \{1,2,\ldots,k\}$ with $i\neq j$. Then, $x_i,x_j\in N(u')\cap I$, a contradiction to $|N(u)\cap I|\leq 1$ for every $u\in C$. We prove that in this case $\rho^o(G)\in \{2,|I|\}$, and the proof is divided into two subcases.\\
	\emph{Case 1.1:} If $k=1 \text{ and }deg(x_1)=1$, then $\rho^o(G)=2$.\\
	Let $N(x_1)=\{y_1\}$. Then, $N(x_1)\cap N(y_1)=\{y_1\}\cap N(y_1)=\emptyset$. So, $\{x_1,y_1\}$ is an open packing of size two in $G$. Thus, $\rho^o(G)\geq 2$. Next, we show that $\rho^o(G)\leq 2$. On the contrary, assume that there exists an open packing $S$ of size at least three in $G$. Then, by Claim~\ref{claim-scapc}, $S\cap C=\emptyset$, i.e., $S\subseteq I$. This implies that $|S|\leq |I|=1$, a contradiction. Therefore, $\rho^o(G)\leq 2$. As we have already proved that $\rho^o(G)\geq 2$ in this case, we can conclude that $\rho^o(G)=2$.\\
	\emph{Case 1.2:} If (i) $k>1$ or (ii) $k=1 \text{ and }deg(x_1)>1$, then $\rho^o(G)=|I|$.\\
	Note that the assumption in Case 1.2 is the negation of Case 1.1. As we have already proved that $N(x_i)\cap N(x_j)=V_i\cap V_j=\emptyset$ for every $x_i,x_j\in I$, $I$ is an open packing in $G$, and hence $\rho^o(G)\geq |I|=k$. Next, we show that $\rho^o(G)\leq |I|=k$. On the contrary, assume that $\rho^o(G)> k$. Let $S$ be an open packing in $G$ of size $\rho^o(G)$. Then, $|S|>|I|=k$, and so $S\cap C\neq \emptyset$. Therefore, by Claim~\ref{claim-scapc}, $|S|\leq 2$. This implies that $2\geq |S|>|I|\geq 1$. So, $|S|=2$ and $|I|=k=1$. Then, $V(G)=C\cup \{x_1\}$, and by the assumption in Case 1.2, $deg(x_1)\geq 2$. Let $S=\{a,b\}$. Since $G$ is constituted of a complete graph ($G[C]$) together with a vertex $x_1$, which is adjacent to at least two vertices in $C$, for any choice of $a,b \in C\cup \{x_1\}$, $a$ and $b$ will have a common neighbour in $G$, a contradiction to $S$ is an open packing of size two in $G$. Therefore, for every open packing $S$ in $G$, $|S|\leq |I|=k$. This implies that $\rho^o(G)\leq |I|$. Since we have proved that $I$ is an open packing in $G$, $\rho^o(G)= |I|$.\\
	\emph{Case 2:} There exists $u\in C$ such that $|N(u)\cap I|= 2$.\\
	In this case, we prove that $\rho^o(G)\in \{1,2\}. $ Note that $|I|\leq 3$ by Lemma~\ref{renjith-lemma}. Let $u\in C$ such that $|N(u)\cap I|= 2$. Let $v,w\in (N(u)\cap I)$. Then, since $|I|\leq 3$, $|I\setminus \{v,w\}|\leq 1$. This implies that for any open packing $S$ in $G$ with $S\subseteq I$, $|S|\leq |S\cap \{v,w\}|+|S\cap (I\setminus \{v,w\})|\leq 2$ because $N(v)\cap N(w)\neq \emptyset$ with $u\in N(v)\cap N(w)$ (and so $|S\cap \{v,w\}|\leq 1$). Also, by Claim~\ref{claim-scapc}, $|S|\leq 2$ if $S\cap C\neq \emptyset$. Since any open packing $S$ in $G$ falls under one of the above two categories (i.e., either $S\subseteq I$ or $S\cap C\neq \emptyset$), $|S|\leq 2$ for every open packing $S$ in $G$. So, $\rho^o(G)\leq 2$. Next, we complete the proof of Case 2 in two cases.\\
	\emph{Case 2.1:} If $\text{ (i) there exist } x,y\in I \text{ such that } N(x)\cap N(y)=\emptyset \text{ or (ii) there exists }  x\in I \text{ such that}$ $deg(x)=1$, then $\rho^o(G)=2$.\\
	If there exist $x,y\in I$ such that $N(x)\cap N(y)=\emptyset$, then $\{x,y\}$ is an open packing of size two in $G$, and so $\rho^o(G)=2$. Also, if there exists a vertex $x\in I$ such that $deg(x)=1$, then it is evident that $n[x]$ is an open packing of size two. Thus, $\rho^o(G)=2$.\\
	\emph{Case 2.2:} $\text{ If (i) } N(x)\cap N(y)\neq \emptyset \text{ for every }x,y\in I,\text{ and (ii) for every }  x\in I,\text{ } deg(x)>1$, then $\rho^o(G)=1$.\\
	Note that the assumption in Case 2.2 is the negation of Case 2.1. We prove that $\rho^o(G)=1$ in this case. As we have already proved that $\rho^o(G)\leq 2$, it is enough to rule out the case $\rho^o(G)=2$. Since $V(G)\subseteq C\cup I$, where (i) any two vertices in $I$ have a common neighbour, (ii) every vertex $x\in I$ has at least two neighbours in $C$, and (iii) $G[C]$ is a complete graph, it is evident that any two vertices in $V(G)$ have a common neighbour in $G$, and so $\rho^o(G)<2$. Therefore, $\rho^o(G)=1$.\\
	Hence, we have shown every possible value of $\rho^o(G)$. Note that every $K_{1,3}$-free split graph $G$ falls under Case 1 or 2. It can be easily verified in polynomial time whether a $K_{1,3}$-free split graph $G$ falls under Case 1 or Case 2. Also, the conditions given in Cases 1.1, 1.2, 2.1, and 2.2 are polynomial time verifiable. Hence, the open packing number of a $K_{1,3}$-free split graph can be found in polynomial time. \end{proof}

\subsection{\boldmath $I_r$-split graphs}\label{sec:ir-split} We observed that a minor modification in Corneil and Perl's reduction \cite{CORNEIL1984} would show that \tdset\ is NP-complete on $I_2$-split graphs. We extended this by proving a dichotomy that \tdset\ in $I_r$-split graphs is (i) NP-complete for $r\geq 2$ and (ii) polynomial time solvable for $r=1$. In addition, we prove that \openpack\ is (i) NP-complete on $I_r$-split graphs for $r\geq 3$ and (ii) polynomial time solvable on $I_r$-split graphs for $r\leq 2$.
We begin the section with Construction~\ref{ik-td} to prove that \tdset\ is NP-complete on $I_r$-split graphs for $r\geq 2$ through a reduction from $r$-\textsc{Hitting Set} problem, which is known to be NP-complete for $r\geq 2$~\cite{karp}.
\begin{construct}
	\emph{Input:} A set $U$ and a set $\mathcal{W}$ of $r$-sized subsets of $U$ for some $r\geq 2$.\\
	\emph{Output:} An $I_r$-split graph $G(C\cup I, E)$.\\
	\emph{Guarantee:} For $k\leq |U|$, $(U,\mathcal{W})$ has a hitting set of size $k$ if and only if $G$ has a total dominating set of size $k+1$.\\
	\emph{Procedure:}\\
	\begin{tabular}{l >{\RaggedRight}p{13cm}}
		Step 1:& For every element $u\in U$, create a vertex $u$ in $G$. Similarly, for every $W\in \mathcal{W}$, create a vertex $z_W$ in $G$.\\
		Step 2:& Create $r+1$ new vertices $\{x_1,x_2,\ldots,x_r,y\}$ in $G$. \\
		Step 3:& Introduce an edge between every pair of distinct vertices of $U\cup \{x_1,x_2,\ldots,x_r\}$ in $G$. \\
		Step 4:& Introduce an edge between a vertex $u\in U$ and $z_W$ for $W\in \mathcal{W}$ in $G$ if $u\in W$.\\
		Step 5:& Add an edge between $x_i$ and $y$ for every $1\leq i\leq r$.\\
	\end{tabular}
	\label{ik-td}
\end{construct}
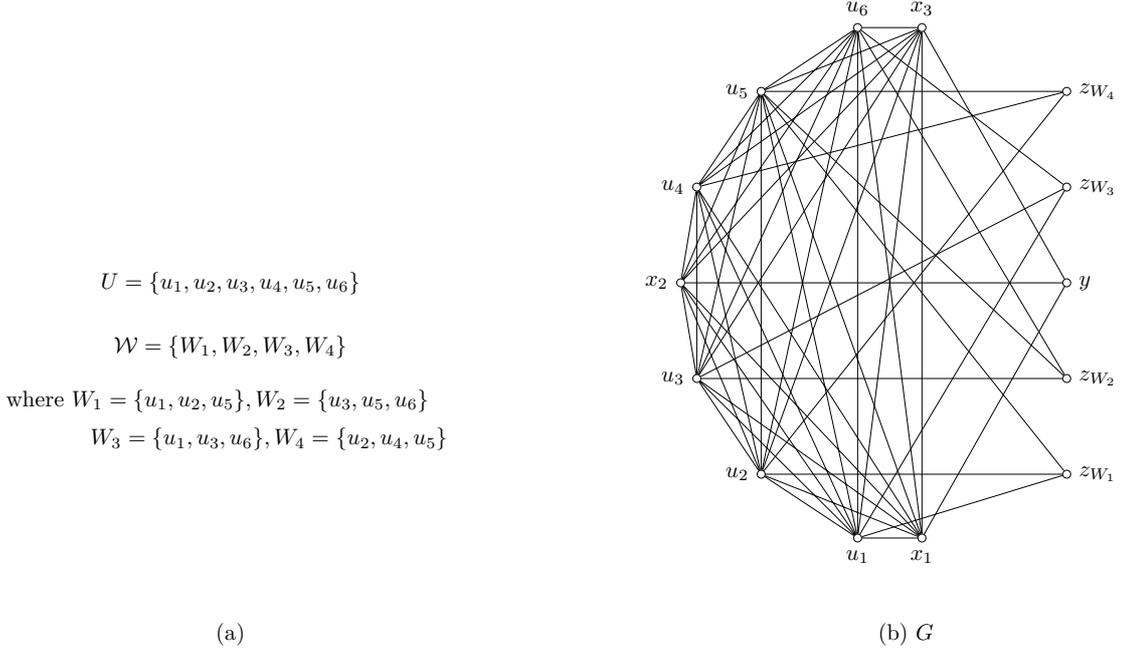
\begin{figure}
	\centering
	\resizebox{\columnwidth}{!}{
		\begin{tikzpicture}[scale=1]
			\draw (-3,0) node{$U=\{u_1,u_2,u_3,u_4,u_5,u_6\}$};
			\draw (-3,-1) node{$\mathcal{W}=\{W_1,W_2,W_3,W_4\}$};
			\draw (-3.2,-1.85)node{where $W_1=\{u_1,u_2,u_5\},W_2=\{u_3,u_5,u_6\}$};
			\draw (-2.4,-2.45)node{$ W_3=\{u_1,u_3,u_6\},W_4=\{u_2,u_4,u_5\}$};
			\draw (6.75,-4)node(u1)[block1,label=below:$u_1$]{};
			\draw (5.25,-3)node(u2)[block1,label=left:$u_2$]{};
			\draw (4.25,-1.5)node(u3)[block1,label=left:$u_3$]{};
			\draw (4.25,1.5)node(u4)[block1,label=left:$u_4$]{};
			\draw (5.25,3)node(u5)[block1,label=left:$u_5$]{};
			\draw (6.75,4)node(u6)[block1,label=above:$u_6$]{};
			\draw (7.75,4)node(x3)[block1,label=above:$x_3$]{};
			\draw (7.75,-4)node(x1)[block1,label=below:$x_1$]{};
			\draw (4,0)node(x2)[block1,label=left:$x_2$]{};
			\draw (u1)--(u2)--(u3)--(u4)--(u5)--(u6)--(u1)--(u3)--(u5)--(u1)--(u4)--(u6)--(u2)--(u5)--(x3)--(x2)--(x1)--(u1)--(x2)--(u2)--(x1)--(u3)--(x2)--(u4)--(x1)--(u5)--(x2)--(u6)--(x1)--(x3)--(u4);
			\draw (u6)--(x3)--(u2);
			\draw (u3)--(x3)--(u1);
			\draw (u2)--(u4);
			\draw (u3)--(u6);
			\draw (10,-3)node(w1)[block1,label=right:$z_{W_1}$]{};
			\draw (10,-1.5)node(w2)[block1,label=right:$z_{W_2}$]{};
			\draw (10,1.5)node(w3)[block1,label=right:$z_{W_3}$]{};
			\draw (10,3)node(w4)[block1,label=right:$z_{W_4}$]{};
			\draw (x1)--(10,0)node(y)[block1,label=right:$y$]{}--(x2);
			\draw (y)--(x3);
			\draw (u1)--(w1)--(u2);
			\draw (u5)--(w1);
			\draw (u3)--(w2)--(u5);
			\draw (u6)--(w2);
			\draw (u1)--(w3)--(u3);
			\draw (u6)--(w3);
			\draw (u4)--(w4)--(u2);
			\draw (u5)--(w4);
			\draw (-3,-5.5)node{(a)};
			\draw (7.5,-5.5)node{(b) $G$};
		\end{tikzpicture}
	}
	\caption{(a) An instance $(U,\mathcal{W})$ of 3-\textsc{Hitting Set} problem and (b) a $I_3$-split graph $G(C\cup I, E)$ produced by Construction~\ref{ik-td}, where $C=U\cup \{x_1,x_2,x_3\}$ is a clique in $G$ and $I=\{z_{W_1},z_{W_2},z_{W_3},z_{W_4},y\}$ is an independent set in $G$ with $deg_G(z_{W_i})=deg_G(y)=3$ for every $i=1,2,3,4$ and $V(G)=C\cup I$. }
	\label{fig-td-hs}
\end{figure}
\noindent An example of Construction~\ref{ik-td} for the case $r=3$ is given in Fig.~\ref{fig-td-hs}. Note that $V(G)=C\cup I$, where $C=U\cup \{x_1,x_2,\ldots,x_r\}$ and $I=\{z_W\,:\,W\in \mathcal{W}\}\cup \{y\}$ and $E(G)=\{uv:u,v\in U\cup \{x_1,x_2,\ldots,x_r\}\}\cup \{uz_W\,:\, u\in U, W\in \mathcal{W} \text{ and }u\in W\}\cup \{x_iy\,:\,1\leq i\leq r\}$. So, $|V(G)|=|U|+|W|+r+1$ and $|E(G)|= \binom{|U|+r}{2}+(r\cdot|W|)+r$. Hence, the graph $G$ can be constructed in polynomial time with respect to $|U|,|\mathcal{W}|$, and $r$. By the construction of $G$, it is clear that $V(G)=C\cup I$ is a clique-independent set partition of the vertex set of $G$. So, $G$ is a split graph. Since the size of every set $W\in \mathcal{W}$ is $r$, $deg_G(z_W)=r$ by the construction of $G$. Further, since $N_G(y)=\{x_1,x_2,\ldots,x_r\}$, $deg_G(y)=r$. This implies that $deg_G(u)=r$ for every $u\in I$. Thus, $G(C\cup I, E)$ is an $I_r$-split graph. \begin{proof}[Proof of Guarantee of Construction~\ref{ik-td}]~\\
	\emph{Case 1:} The instance $(U,\mathcal{W})$ of $r$-\textsc{Hitting Set} has a hitting set of size $k$.\\
	Let $U'=\{u_1,u_2,\ldots,u_k\}$ be a hitting set of size $k$ in $(U,\mathcal{W})$. Then, we show that $\{u_1,u_2,\ldots,u_k,x_1\}$ is a total dominating set of size $k+1$ in $G$. By the construction of $G$, every vertex in $C\cup \{y\}$ is adjacent to $u_1$ or $x_1$. So, to prove the claim, it is enough to show that every vertex in $(I\setminus\{y\})=\{z_W\,:\,W\in \mathcal{W}\}$ is adjacent to some vertex in $\{u_1,u_2,\ldots,u_k,x_1\}$. Let $W\in \mathcal{W}$. Then, since $\{u_1,u_2,\ldots,u_k\}$ is a hitting set of $(U,\mathcal{W})$, there exists $i\in \{1,2,\ldots,k\}$ such that $u_i\in W$. Hence, by the construction of $G$, $z_W$ is adjacent to $u_i$. Since $W$ is an arbitrary set in $\mathcal{W}$, for every $W\in \mathcal{W}$, the corresponding vertex $z_W$ in $G$ is adjacent to some vertex in $\{u_1,u_2,\ldots,u_k\}$. Thus, every vertex in $V(G)=C\cup I$ is adjacent to some vertex in $\{u_1,u_2,\ldots,u_k,x_1\}$, and hence $\{u_1,u_2,\ldots,u_k,x_1\}$ is a total dominating set of size $k+1$ in $G$.\\
	\emph{Case 2:} $G$ has a total dominating set of size $k+1$.\\
	Let $D$ be a total dominating set of size $k+1$ in $G$. Then, $D\cap \{x_1,x_2,\ldots,x_r\}\neq \emptyset$ because $N_G(y)=\{x_1,x_2,\ldots,x_r\}$ and $D\cap N_G(y)\neq \emptyset$. Let $D'=D\cap U$. Since  $D\cap \{x_1,x_2,\ldots,x_r\}\neq \emptyset$ and $U\cap  \{x_1,x_2,\ldots,x_r\}= \emptyset$, $|D'|=|D\cap U|\leq k$. Next, we show that $D'$ is a hitting set of size at most $k$ of $(U,\mathcal{W})$. Let $W\in \mathcal{W}$. Then, since $D$ is a total dominating set, there exists a vertex $u\in D$ such that $uz_W\in E(G)$. Also, since $N_G(z_W)\subseteq U$ by the construction of $G$, $u\in D \cap U=D'$. So, by the construction of $G$, $u\in W$. Since $W$ is an arbitrary set in $\mathcal{W}$, for every $W\in \mathcal{W}$, there exists an element $u\in D'$ such that $u\in W$. Hence, $D'$ is a hitting set of size at most $k$ in $(U,\mathcal{W})$. 
\end{proof}
\noindent The theorem below follows from Construction~\ref{ik-td} and the fact that $r$-\textsc{Hitting Set} is NP-complete for $r\geq 2$~\cite{karp}.
\begin{theorem}
	\tdset\ is NP-complete on $I_r$-split graphs for $r\geq 2$.
\end{theorem}
\noindent The following theorem completes the dichotomy result for \tdset\ on $I_r$-split graphs.
\begin{theorem}
	Let $G(C\cup I, E)$ be a connected non-trivial $I_1$-split graph. Then, $\gamma_t(G)=\max\{2,|N(I)|\}$, where $N(I)=\bigcup\limits_{u\in I}N(u)$. Hence, an optimal total dominating set of $G$ can be found in linear time.
	\label{thm-td-I_1}
\end{theorem}
\begin{proof}
	If $I=\emptyset$, then $G$ is a complete graph, and so $\gamma_t(G)=2$. Thus, assume that $I\neq \emptyset$. Note that by the definition of $I_1$-split graphs, every vertex in $I$ is a pendant vertex in $G$. Since $N(u)\cap D$ is non-empty for every total dominating set $D$ in $G$, $N(u)\subseteq D$ for every pendant vertex $u$ in $G$. This implies that $N(I)\subseteq D$ for every total dominating set $D$ in $G$, and so $|N(I)|\leq \gamma_t(G)$. We prove the theorem in two cases based on $|N(I)|$.\\
	\emph{Case 1:} $|N(I)|=1$.\\
	Let $N(I)=\{u\}$. Since $I$ is an independent set in $G$, $u\in C$. Since $C$ is a clique and $N(I)=\{u\}$, every vertex in $V(G)\setminus\{u\}=(C\cup I)\setminus \{u\}$ is adjacent to $u$. Let $v\in I$. Then, since $N(I)=\{u\}$, $vu\in E(G)$. This implies that every vertex in $V(G)=(V(G)\setminus\{u\})\cup \{u\}$ is adjacent to a vertex in $\{u,v\}$. So, $\{u,v\}$ is a total dominating set in $G$, and hence $\gamma_t(G)\leq 2$. Since $\gamma_t(H)\geq 2$ for every graph $H$ that admits a total dominating set, $\gamma_t(G)=2$.\\
	\emph{Case 2:} $|N(I)|\geq 2$.\\
	Since $I$ is an independent set in $G$, $N(I)\subseteq C$. Let $u,u'\in N(I)$. Then, since $C$ is a clique and $u,u'\in N(I)\subseteq C$, every vertex in $C$ is adjacent to $u$ or $u'$. Note that by the definition of $N(I)$, every vertex in $I$ is adjacent to some vertex in $N(I)$. Thus, every vertex in $V(G)=C\cup I$ is adjacent to a vertex in $N(I)$, and so $N(I)$ is a total dominating set. This implies that $\gamma_t(G)\leq |N(I)|$. Recall that $\gamma_t(G)\geq |N(I)|$. Hence, $\gamma_t(G)=|N(I)|$. \\
	Thus, by Cases 1 and 2, $\gamma_t(G)=\max\{2,|N(I)|\}$. Note that $|N(I)|$ can be found in linear time. So, an optimal total dominating set of $G$ can be found in linear time.
\end{proof}
\noindent Next, we prove the dichotomy result on \openpack\ in $I_r$-split graphs. Firstly, we show that \openpack\ is NP-complete on $I_r$-split graphs for $r\geq 3$ through a reduction from $r$-\textsc{Dimensional Matching} problem, which is known to be NP-complete for $r\geq 3$~\cite{karp}.
\begin{construct}
	\emph{Input:} A collection of sets $X_1,X_2,\ldots,X_r$ such that $|X_i|=q$ for $i=1,2,\ldots, r$, for some $q\in \mathbb{N}$, and a non-empty set $M\subseteq \prod\limits_{i=1}^r X_i$.\\
	\emph{Output:} An $I_r$-split graph $G(C\cup I, E)$.\\
	\emph{Guarantee:} $(X_1,X_2,\ldots,X_r,M)$ is a yes-instance of $r$-\textsc{Dimensional Matching} if and only if $G$ has an open packing of size $q$.\\
	\emph{Procedure:}\\
	\begin{tabular}{l >{\RaggedRight}p{13.25cm}}
		Step 1:& For every $1\leq i\leq r$ and for every $x\in X_i$, create a vertex $z_{(x,i)}$ in $G$. Similarly, for $w\in M$, create a vertex $y_w$ in $G$.\\
		Step 2:& Introduce an edge between every pair of distinct vertices of $C=\cup_{i=1}^r\{z_{(x,i)}\,:\,x\in X_i\}$ in $G$. \\
		Step 3:& Introduce an edge between the vertex $z_{(x,i)}\in C$ and $y_w$ in $G$ if the $i^{th}$ coordinate of $w$ is $x$.
	\end{tabular}
	\label{ir-op}
\end{construct}
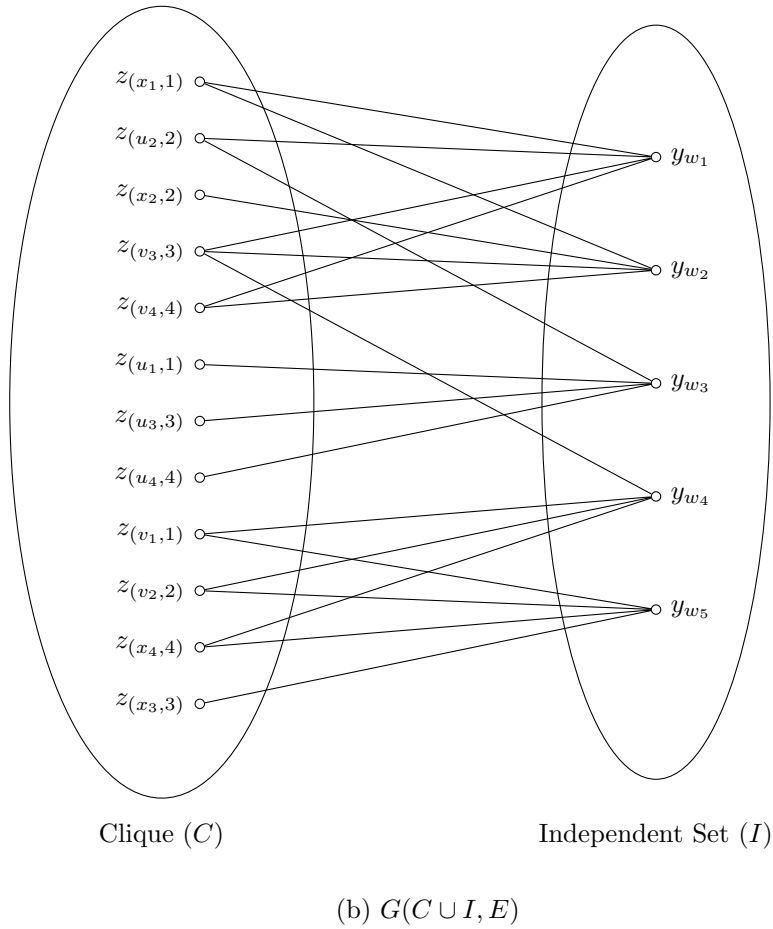
\begin{figure}
	\centering
	\begin{tikzpicture}
		\draw (-5,0) node{$X_1=\{x_1,u_1,v_1\}$};
		\draw (-1.75,0)node{$X_2=\{x_2,u_2,v_2\}$};
		\draw (1.5,0)node{$X_3=\{x_3,u_3,v_3\}$};
		\draw (4.75,0)node{$X_4=\{x_4,u_4,v_4\}$};
		\draw (-0.5,-1) node{$M=\{w_1,w_2,w_3,w_4,w_5\}$};
		\draw (-4.25,-1.85)node{where $w_1=(x_1,u_2,v_3,v_4)$};
		\draw (3.7,-1.85)node{$w_3=(u_1,u_2,u_3,u_4)$};
		\draw (-3.7,-2.45)node{$ w_4=(v_1,v_2,v_3,x_4)$};
		\draw (0,-2.45)node{$w_5=(v_1,v_2,x_3,x_4)$};
		\draw (0,-1.85)node{$w_2=(x_1,x_2,v_3,v_4)$};
		
		\draw (-3,-6)node(x11)[block1,label=left:$z_{(x_1,1)}$]{};
		\draw (-3,-9.75)node(x12)[block1,label=left:$z_{(u_1,1)}$]{};
		\draw (-3,-12)node(x13)[block1,label=left:$z_{(v_1,1)}$]{};
		\draw (-3,-7.5)node(x21)[block1,label=left:$z_{(x_2,2)}$]{};
		\draw (-3,-6.75)node(x22)[block1,label=left:$z_{(u_2,2)}$]{};
		\draw (-3,-12.75)node(x23)[block1,label=left:$z_{(v_2,2)}$]{};
		
		\draw (-3,-14.25)node(x31)[block1,label=left:$z_{(x_3,3)}$]{};
		\draw (-3,-10.5)node(x32)[block1,label=left:$z_{(u_3,3)}$]{};
		\draw (-3,-8.25)node(x33)[block1,label=left:$z_{(v_3,3)}$]{};
		\draw (-3,-13.5)node(x41)[block1,label=left:$z_{(x_4,4)}$]{};
		\draw (-3,-11.25)node(x42)[block1,label=left:$z_{(u_4,4)}$]{};
		\draw (-3,-9)node(x43)[block1,label=left:$z_{(v_4,4)}$]{};
		\draw (-3.5,-10.25)ellipse[x radius=2, y radius=5.25];
		\draw (3,-7)node(y1)[block1,label=right:$y_{w_1}$]{};
		\draw (3,-10)node(y2)[block1,label=right:$y_{w_3}$]{};
		\draw (3,-11.5)node(y3)[block1,label=right:$y_{w_4}$]{};
		\draw (3,-13)node(y4)[block1,label=right:$y_{w_5}$]{};
		\draw (3,-8.5)node(y5)[block1,label=right:$y_{w_2}$]{};
		\draw (3,-10.25)ellipse[x radius=1.5, y radius=5];
		\draw (x11)--(y1)--(x22);
		\draw (x33)--(y1)--(x43);
		
		\draw (x12)--(y2)--(x22);
		\draw (x32)--(y2)--(x42);
		
		\draw (x13)--(y3)--(x23);
		\draw (x33)--(y3)--(x41);
		
		\draw (x13)--(y4)--(x23);
		\draw (x31)--(y4)--(x41);
		
		\draw (x11)--(y5)--(x21);
		\draw (x33)--(y5)--(x43);
		\draw (-3.5,-16)node{Clique ($C$)};
		\draw (3,-16)node{Independent Set ($I$)};
%
		\draw (0,-4)node{(a) $4$-\textsc{Dimensional Matching}};
		\draw (0,-17)node{(b) $G(C\cup I, E)$};
	\end{tikzpicture}
	\captionsetup{format=hang}
	\caption{(a) An instance $(X_1,X_2,X_3,X_4,M)$ of $4$-\textsc{Dimensional Matching} and (b) An $I_r$-split graph $G(C\cup I,E)$ corresponding to the instance $(X_1,X_2,X_3,X_4,M)$ produced using Construction~\ref{ir-op}. The edges between the vertices of $C$ aren't given in the figure for clarity. Note that for every $X_i=\{z_i,u_i,v_i\}$ in the instance of $4$-\textsc{Dimensional Matching}, three vertices $z_{(x_i,i)},z_{(u_i,i)}$ and $z_{(v_i,i)}$ are created in the graph $G$.}
	\label{fig:ir-op-construct}
\end{figure}
\noindent An example of Construction~\ref{ir-op} is given in Fig.~\ref{fig:ir-op-construct}. The vertex and edge sets of the graph $G$ are  $V(G)=C\cup I$, where $I=\{y_w\,:\,w\in M\}$ and $E(G)=\{xx'\,:\,x,x'\in C\}\cup \{z_{(x,i)}y_w\,:\,x\in X_i, w\in M,\text{ and the }i^{th} \text{ coordinate of }w\text{ is }x \}$. Then, $|V(G)|=(r\cdot q)+|M|$ and $|E(G)|=r\cdot|M|+\binom{r\cdot q}{2}$. Hence, the graph $G$ can be constructed in polynomial time with respect to $r, q,$ and $|M|$. Note that by the construction of $G$, $C$ is a clique and $I$ is an independent set, and hence $V(G)=C\cup I$ is a clique-independent set partition of $V(G)$. Thus, $G(C\cup I,E)$ is a split graph. Further, by the construction of $G$, it is clear that for $y_w\in I$, $N_G(y_w)=\{z_{(x_1,1)},z_{(x_2,2)},\ldots,z_{(x_r,r)}\}$ if $w=(x_{1},x_{2},\ldots,x_{r})$. Hence, $deg_G(y_w)=r$ for every $y_w\in I$. Thus, $G$ is an $I_r$-split graph. Next, we prove the guarantee of Construction~\ref{ir-op}. The following claim is used in the proof of guarantee of Construction~\ref{ir-op}.
 \begin{claim-alt}
	If $S$ is an open packing in the output graph $G$ of Construction~\ref{ir-op} such that $S\cap C$ is non-empty, then $|S|=1$.\\
	 On the contrary, assume that there exists an open packing $S$ in $G$ such that $S\cap C$ is non-empty and $|S|>1$. Let $u\in S\cap C$, and let $z\in S\setminus \{u\}$. Since $r\geq 3$ and $q\geq 1$, $|C|\geq 3$. This implies that $C$ is a clique of size at least three in $G$, and so any two vertices in $C$ have a common neighbour in $G$. This says that $z\notin C$. So, $z\in I$, i.e., $z=y_w$ for some $w\in M$. Since $r\geq 3$ and $deg_G(y_w)=r$ with $N_G(y_w)\subseteq C$, $(z=)y_w\in I$ is adjacent to at least two vertices other than $u$ in $C$. This implies that $N_G(y_w)\cap N_G(u)\neq \emptyset$, a contradiction to $u,z(=y_w)\in S$ and $S$ is an open packing in $G$.  
	\label{claim-info-xcaps}
 \end{claim-alt}
\begin{proof}[Proof of Guarantee of Construction~\ref{ir-op}] We prove the Guarantee of Construction~\ref{ir-op} in two cases.\\
	\emph{Case 1:} $(X_1,X_2,\ldots,X_r,M)$ is a yes-instance of $r$-\textsc{Dimensional Matching}.\\
	W.L.O.G., assume that $\{w_1,w_2,\ldots,w_q\}$ is an $r$-\textsc{Dimensional Matching} of $(X_1,X_2,$ $\ldots,X_r,M)$. Then, for $w_j=(x_1,x_2,\ldots,x_r)$ and $w_p=(u_1,u_2,\ldots,u_r)$ with $p,j\in \{1,2,\ldots,q\}$ and $p\neq j$, $x_i\neq u_i$ for $i\in \{1,2,\ldots,r\}$. Therefore, by the construction of $G$, $z_{(x_t,t)}\neq z_{(u_\ell,\ell)}$ for every $t,\ell\in \{1,2,\ldots,r\}$. This implies that $N_G(y_{w_j})\cap N_G(y_{w_p})=\{z_{(x_1,1)},z_{(x_2,2)},\ldots,z_{(x_r,r)}\}\cap \{z_{(u_1,1)},z_{(u_2,2)},\ldots,z_{(u_r,r)}\}$ $=\emptyset$. So, $\{y_{w_1},y_{w_2},\ldots,y_{w_q}\}$ is an open packing of size $q$ in $G$.\\
	\emph{Case 2:} $G$ has an open packing $S$ of size $q$.\\
	We further divide this case into two cases based on $S\cap C$.\\
	\emph{Case 2.1:} $S\cap C$ is non-empty.\\
	Then, by Claim~\ref{claim-info-xcaps}, $|S|=q=1$. Therefore, $|X_i|=1$ for every $i\in \{1,2,\ldots,r\}$. Let $X_i=\{x_i\}$ for $i=1,2,\ldots,r$. Since $M$ is non-empty, there exists exactly one $r$-tuple $w\in M$ such that $w=(x_1,x_{2},\ldots,x_{r})$. Hence, $M=\{w\}$ is an $r$-\textsc{Dimensional Matching} of $(X_1,X_2,\ldots,X_r,M)$.\\
	\emph{Case 2.2:} $S\cap C$ is empty.\\
	So, $S\subseteq I$. W.L.O.G., assume that $S=\{y_{w_1},y_{w_2},\ldots,y_{w_q}\}$. Then, we show that $L=\{w_1,w_2,\ldots,$ $w_q\}$ is an $r$-\textsc{Dimensional Matching} of $(X_1,X_2,\ldots,X_r,M)$. On the contrary, assume that there exist $j,p\in \{1,2,\ldots,q\}$ and  $i\in \{1,2,\ldots,r\}$ such that $x_i= u_i$ for $w_j=(x_1,x_2,\ldots,x_r)$ and $w_p=(u_1,u_2,\ldots,u_r)$. Then, by the construction of $G$, $z_{(x_i,i)}= z_{(u_i,i)}\in N_G(y_{w_j})\cap N_G(y_{w_p})$, a contradiction to $y_{w_j},y_{w_p}\in S$ and $S$ is an open packing in $G$. So, $L=\{w_1,w_2,\ldots,w_q\}$ is an $r$-\textsc{Dimensional Matching} of $(X_1,X_2,\ldots,X_r,M)$. 
\end{proof}
\begin{theorem}
	For $r\geq 3$, \openpack\ is NP-complete on $I_r$-split graphs.
\end{theorem}
\begin{proof}
	Given an $I_r$-split graph $G$ and a vertex subset $S$ of $G$, it can be tested in linear time whether $S$ is an open packing in $G$ or not using Lemma~\ref{thm-openpacking-testing}. Hence, \openpack\ is in the class NP on $I_r$-split graphs. Further, Construction~\ref{ir-op} and the fact that $r$-\textsc{Dimensional Matching} is NP-complete for $r\geq 3$~\cite{karp} implies that \openpack\ is NP-complete on $I_r$-split graphs for $r\geq 3$.
\end{proof}
\noindent The following construction is used to show that \openpack\ is polynomial time solvable on a superclass of $I_1$-split graphs and $I_2$-split graphs. Thus, completing our dichotomy result on $I_r$-split graphs. 
\begin{construct}
	\emph{Input:} A split graph $G(C\cup I,E)$ with $1\leq deg_G(v)\leq 2$ for every $v\in I$.\\
	\emph{Output:} A graph $G'(V',E')$.\\
	\emph{Guarantee:} For $k\in \mathbb{N}$, $G$ has an open packing $S$ such that $S\subseteq I$ and $|S|=k$ if and only if $G'$ has a matching of size $k$.\\
	\emph{Procedure:}\\
	\begin{tabular}{l >{\RaggedRight}p{13.25cm}}
		Step 1:& For every vertex $u\in C$ of $G$, create a vertex $u$ in $G'$.\\[2pt]
		Step 2:& For every vertex $v\in I$ with $deg_G(v)=2$ and $N_G(v)=\{u,u'\}$, create an edge $e_v$ with $e_v=uu'$ in $G'$.\\[2pt]
		Step 3:& For every vertex $v\in I$ with $deg_G(v)=1$ and $N_G(v)=\{u\}$, create a vertex $v$ and an edge $e_v$ with $e_v=uv$ in $G'$.
	\end{tabular}
	\label{op-match-2}
\end{construct}
\noindent Let $J=\{v\in I\,:\,deg_G(v)=1\}$. Then, the vertex set and the edge set of $G'$ are $V'=C\cup J$ and $E'=\{e_v\,:\,v\in I\}$. Note that $|V'|\leq n$ and $|E'|\leq n$. Since $G'$ is constructed based on $N_G(v)$ for every $v\in I$, it can be constructed in $O(n+m)$ time. Also, note that the graph $G'$ may have parallel edges. Next, we prove the guarantee of the above construction.
\begin{figure}
	\centering
	\begin{tikzpicture}
		\draw (0,0)node(x2)[block1,label=left:$x_2$]{}--(1,-1.5)node(x1)[block1,label=left:$x_1$]{}--(1,3)node(x4)[block1,label=left:$x_4$]{}--(0,1.5)node(x3)[block1,label=left:$x_3$]{}--(x2)--(x4);
		\draw (x3)--(x1);
		\draw (x1)--(3,-1.5)node(u)[block1,label=right:$u$]{}--(x2);
		\draw (x2)--(3,0.75)node(w)[block1,label=right:$w$]{}--(x3);
		\draw (x3)--(3,1.875)node(y)[block1,label=right:$y$]{};
		\draw (x4)--(3,3)node(z)[block1,label=right:$z$]{};
		\draw (x1)--(3,-.375)node(v)[block1,label=right:$v$]{}--(x2);
		\draw (0.25,.75)ellipse[x radius=1.5, y radius=3];
		\draw (3.25,.75)ellipse[x radius=.75, y radius=3];
		\draw (0.25,-2.75)node{$C$};
		\draw (3.25,-2.75)node{$I$};
		\draw (2,-3.5)node{(a) $ G(C\cup I,E)$};
		\draw (8,.5)node(x2)[block1,label=left:$x_2$]{}..controls(7.5,-.75)and(7.5,-.25)..(8,-1.5)node(x1)[block1,label=left:$x_1$]{};
		\draw (x2)..controls(8.5,-.75)and(8.5,-0.25)..(x1);
		\draw (x2)--(8,2)node(x3)[block1,label=left:$x_3$]{}--(9.5,2)node(y)[block1,label=right:$y$]{};
		\draw (8,3)node(x4)[block1,label=left:$x_4$]{}--(9.5,3)node(z)[block1,label=right:$z$]{};
		\draw (7.25,-.5)node{$e_u$};
		\draw (8.75,-.5)node{$e_v$};
		\draw (7.75,1.25)node{$e_w$};
		\draw (8.75,2.25)node{$e_y$};
		\draw (8.75,3.25)node{$e_z$};
		\draw (8.75,-3.5)node{(b) $ G'(V',E')$};
	\end{tikzpicture}
	\captionsetup{format=hang}
	\caption{(a) A split graph $G(C\cup I,E)$ with $1\leq deg_G(v)\leq 2$ for every $v\in I$ and (b) the graph $G'(V',E')$ corresponding to $G$ constructed by Construction~\ref{op-match-2}.}
\end{figure}
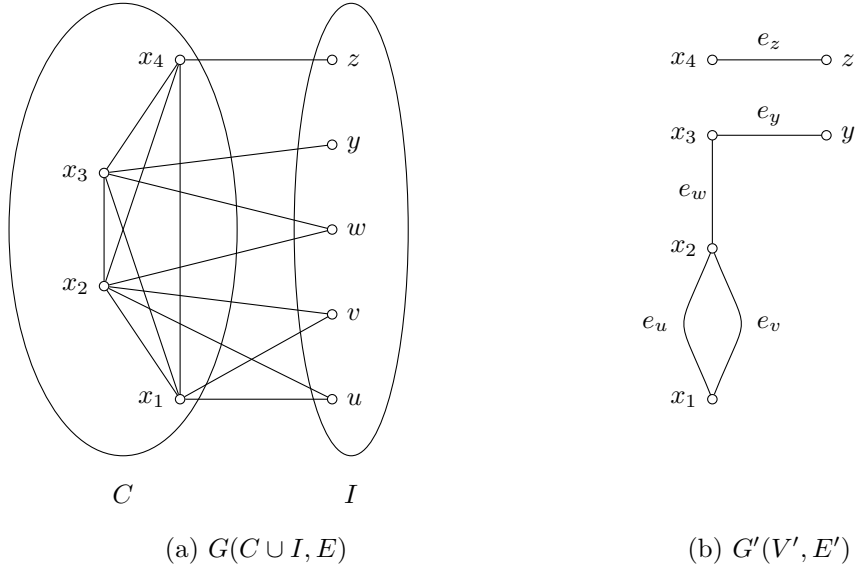
\begin{proof}[Proof of Guarantee of Construction~\ref{op-match-2}] We prove the guarantee of the construction in two cases.\\
	\emph{Case 1:} $G$ has an open packing $S$ such that $|S|=k$ and $S\subseteq I$.\\
	Let $F_S=\{e_v\in E'\,:\,v\in S\}$. We show that $F_S$ is a matching in $G'$. On the contrary, assume that $F_S$ is not a matching in $G'$. Then, there exist two distinct edges $e_v,e_y\in F_S$ such that $V_{G'}(e_v)\cap V_{G'}(e_y)\neq \emptyset$. Let $u\in V_{G'}(e_v)\cap V_{G'}(e_y)$, i.e., $e_v,e_y\in E_{G'}(u)$. Then, by the construction of $G'$, $u\in C\cup J$ and $v,y\in I$. Firstly, we show that $u\notin J$. If not, $u\in J$. Then, $deg_G(u)=1$, and so, by Step 3 of Construction~\ref{op-match-2}, $E_{G'}(u)=\{e_u\}$. This implies that $e_v=e_u=e_y$, a contradiction to $v\neq y$. So, $u\notin J$. Hence, $u\in C$, and so by the construction of $G'$, $uv,uy\in E(G)$. This implies that $u\in N_G(v)\cap N_G(y)$, a contradiction to $v,y\in S$ and $S$ is an open packing in $G$. So, $F_S$ is a matching of size $k$ in $G'$.\\
	\emph{Case 2:} $F$ is a matching in $G'$ of size $k$.\\
	Let $S_F=\{v\in I\,:\,y_v\in F\}$. In this case, we prove that $S_F$ is an open packing in $G$. On the contrary, assume that $S_F$ is not an open packing in $G$. Then, there exist vertices $v,y\in S_F\subseteq I$ such that $N_G(v)\cap N_G(y)\neq \emptyset$. Let $u\in (N_G(v)\cap N_G(y))\subseteq  C$. This implies that $u\in V(G')=C\cup J$. Also, by the construction of $G'$, $u\in V_{G'}(e_v)\cap V_{G'}(e_y)$, a contradiction to $e_v,e_y\in F$ and $F$ is a matching in $G'$. Hence, $S_F$ is an open packing in $G$ such that $S\subseteq I$ and $|S|=k$. 
\end{proof}
\noindent The following lemma and Construction~\ref{op-match-2} help us compute a maximum open packing in $I_1$-split graphs and $I_2$-split graphs.
\begin{lemma}[\cite{Micali80}]
	A maximum matching in a graph with $n$ vertices and $m$ edges can be found in $O(m\sqrt{n})$ time.\label{micali}\end{lemma}
\begin{theorem}
	A maximum open packing in a split graph $G(C\cup I,E)$ with $1\leq deg_G(v)\leq 2$ for every $v\in I$ can be found in $O(\max\{n+m,n^{1.5}\})$ time.
	\label{thm-poly-i1,i2}
\end{theorem}
\begin{proof}
	Let $G(C\cup I, E)$ be a split graph with $deg_G(v)\in \{1,2\}$ for every $v\in I$. Also, let $F$ be a maximum matching of the graph $G'$ constructed in Construction~\ref{op-match-2}, i.e., $\alpha'(G')=|F|$. Then, we find the open packing number of $G$ using $F$ in four parts.\\
	\emph{Part 1:} If $\alpha'(G')\geq 3$, then $S_F=\{v\in I\,:\, e_v\in F\}$ is a maximum open packing in $G$.\\
	Note that $S_F$ is an open packing in $G$ by Case 2 in the proof of Guarantee of Construction~\ref{op-match-2}. So, $\rho^o(G)\geq |S_F|=\alpha'(G)\geq 3$. Hence, by Claim~\ref{claim-scapc}, every maximum open packing in $G$ is a subset of $I$. Let $S$ be a maximum open packing in $G$. Then, by Case 1 in the proof of Guarantee of Construction~\ref{op-match-2}, $F_S=\{e_v\in E'\,:\,v\in S\}$ is a matching in $G'$. So, $\alpha'(G')\geq |F_S|=|S|=\rho^o(G)\geq |S_F|=\alpha'(G')$. Hence, $\rho^o(G)=\alpha'(G')$, and $S_F$ is a maximum matching in $G'$.\\
	\emph{Part 2:} If $\alpha'(G')=2$, then $S_F=\{v\in I\,:\, e_v\in F\}$ is a maximum open packing in $G$.\\
	If $\rho^o(G)\geq 3$, then by Claim~\ref{claim-scapc} and Construction~\ref{op-match-2}, $\alpha'(G')\geq 3$, a contradiction to the assumption that $\alpha'(G')=2$. So, $\rho^o(G)\leq 2$. By Case 2 in the proof of Guarantee of Construction~\ref{op-match-2}, $S_F$ is an open packing of size two in $G$. Hence, $\rho^o(G)=2$.\\
	\emph{Part 3:} If $\alpha'(G')=1$ and there exists a vertex of degree one in $G$, then $\rho^o(G)=2$.\\
	Since $\alpha'(G')=1$, it is evident, by Claim~\ref{claim-scapc} and Construction~\ref{op-match-2}, that $\rho^o(G)\leq 2$. Let $v\in V(G)$ such that $deg_G(v)=1$. Then, $N_G[v]$ is an open packing of size two in $G$. Hence, $\rho^o(G)=2$.\\
	\emph{Part 4:} If $\alpha'(G')=1$ and every vertex in $G$ is of degree at least two, then $\rho^o(G)=1$.\\
	On the contrary, assume that $G$ has an open packing $S$ of cardinality at least two. Since $\alpha'(G')=1$, there does not exist an open packing $S$ in $G$ such that $|S|\geq 2$ and $S\subseteq I$ by Construction~\ref{op-match-2}. So, $S\cap C\neq \emptyset$. Let $u\in S\cap C$. Then, by Claim~\ref{claim-scapc}, $S\subseteq N_G[u]$. Let $v\in S\setminus\{u\}$. Then, by the assumption on Part 4 that every vertex in $G$ is of degree at least two, for any choice of $v$ in $V(G)\setminus \{u\}=(C\cup I)\setminus\{u\}$, $u$ and $v$ must have a common neighbour in $G$, a contradiction to $S$ being an open packing in $G$ with $u,v\in S$. So, $\rho^o(G)=1$.\\
	Parts 1-4 complete the arguments for finding a maximum open packing in $G$. Next, we look at the complexity of computing it. Note that
	\begin{itemize}
		\item[(a)] the graph $G'$ can be constructed in $O(n+m)$ time,
		\item[(b)] a maximum matching in $G'$ can be found in $O(|E(G')|\sqrt{V(G')})=O(n^{1.5})$ time using Lemma~\ref{micali}, and 
		\item[(c)] degree of every vertex in $G$ can be found in $O(n+m)$ time.
	\end{itemize} 
	\noindent By (a)-(c) and Parts 1-4, a maximum open packing in $G$ can be found in $O(\max\{n+m,n^{1.5}\})$ time.
\end{proof}
\section{Conclusion}
In this article, we have completed the study on the complexity (P vs NPC) of \openpack\ on $H$-free graphs for every graph $H$ with at least three vertices by proving that \openpack\ is (i) NP-complete on $K_{1,3}$-free graphs and (ii) polynomial time solvable on subclasses of $(P_4\cup rK_1)$-free graphs for every $r\geq 1$. Also, we proved that for a connected $(P_4\cup rK_1)$-free graph $G$, $\gamma_t(G)\leq 2r+2$ and $\rho^o(G)\leq 2r+1$. Since $\rho^o(G)\leq \gamma_t(G)$, it would be interesting to study (i) the bound on the difference between $\gamma_t(G)$ and $\rho^o(G)$ and (ii) the lower bounds for open packing number. Further, we proved that \openpack\ is (i) NP-complete on $K_{1,4}$-free split graphs and (ii) polynomial time solvable on $K_{1,3}$-free split graphs. We also showed that \openpack\ is (i) NP-complete on $I_r$-split graphs for $r\geq 3$ and (ii) polynomial time solvable on $I_r$-split graphs for $r\leq 2$. Also, we completed the complexity (P vs NPC) of \tdset\ on $I_r$-split graphs by proving that it is NP-complete when $r\geq 2$ and polynomial time solvable when $r=1$. Note that $I_r$-split graphs is a split graph with regularity (equal number of edges incident) on the vertices of the independent set part of the clique-independence partition of split graphs. So, it would be interesting to know the complexity (P vs NPC) of \tdset\ and \openpack\ in split graphs with regularity imposed on (i) the clique part and (ii) both the clique and the independent set parts of the clique-independence partition of split graphs, where the regularity on clique part shall be different from the regularity on independent set part.

\bibliographystyle{splncs04}
\bibliography{openpackingarxiv}

\newpage
\appendix

\begin{subappendices}
 \section{Open Packing}\label{app-sec:op-test}
In this section, we give an algorithm to test whether a subset $S$ of a graph $G$ is an open packing in $G$ or not.
\begin{algorithm}
	\DontPrintSemicolon
	\SetAlgoLined
	\SetKwInOut{Input}{Input}\SetKwInOut{Output}{Output}
	\Input{A Graph $G(V,E)$ and a vertex subset $S$ of $G$.}
	\Output{ $S$ is/isn't an open packing in $G$}
	
	\BlankLine
	
	\textbf{Initialization:} $U = S$, $ W= \emptyset$, an empty array $Q$, $j=1$ and every vertex in $V(G)$ is unmarked.\;
	{
		\While{$U\neq \emptyset$}
		{
			Choose a vertex $u\in U$\;
			Set $W=N_G(u)$\;
			\While{$W\neq \emptyset$}
			{
				Choose a vertex $v\in W$\;
				Set $Q[j]=v$ \quad  $\backslash\backslash$ $Q$ is used only in the analysis of the algorithm\;
				Set $j=j+1$\;
				\uIf{$v$ is marked}{return `$S$ is not an open packing in $G$' and terminate\;}
				\Else{Mark $v$\;
					Set $W=W\setminus\{v\}$\;}
			}
			Set $U=U\setminus\{u\}$\;
		}	
		Return `$S$ is an open packing in $G$' and terminate\; 
	}
	\caption{Open Packing Testing Algorithm for Graphs}
	\label{testing-algo-for-op}
\end{algorithm} 
\begin{observation}
	Given a graph $G$ and a vertex subset $S$ of $G$, Algorithm~\ref{testing-algo-for-op} determines whether $S$ is an open packing in $G$ or not in $O(n)$ time. 
	\label{openpacking-testing}
\end{observation}
\noindent Lemma~\ref{thm-openpacking-testing} follows from Observation~\ref{openpacking-testing}.
\begin{proof}[Proof of Observation~\ref{openpacking-testing}] We start by giving some simple clarifications to get a clear understanding of our proof. We use the adjacency list of the graph $G$ as the data structure. We say that a vertex $u\in V(G)$ is visited by the algorithm if the vertex is chosen by at least one of the while loops during some stage of the algorithm. We call the while loop that is performed based on whether $U$ is empty or not the external while loop. We refer to the iterations performed by the external while loop as the iterations of the algorithm. The while loops that are performed based on the set $W$ are called the internal while loops. 
	
	\noindent We begin with the proof for the termination of the algorithm. Note that if the algorithm visits some vertex that is already marked, then the algorithm will terminate. Otherwise, the algorithm visits no marked vertex in any iteration of the algorithm. However, since $W$ is a finite set for every iteration of the external while loop, the internal while loop iterates only a finite number of times. Hence, every iteration of the external while loop will terminate. Also, since the size of $U$ is finite, the number of iterations performed by the algorithm is finite, and hence, the algorithm would terminate once $U$ becomes empty.
	
	\noindent Next, we prove the correctness of Algorithm~\ref{testing-algo-for-op}. Let $S=\{u_1,u_2,\ldots,u_k\}$. W.L.O.G., assume that the algorithm chooses the vertices in $S$ in the increasing order of their suffices. We prove the correctness of the algorithm in two cases based on whether the input vertex subset $S$ is an open packing in $G$ or not.\\
	\emph{Case 1:} The input vertex subset $S$ is an open packing in $G$.\\
	Since every vertex in $V(G)$ is unmarked during the initialization process, no vertex in $N_G(u_1)$ might be marked while examining them. For $2\leq i\leq k$, the vertices of $G$ that are marked during the previous iterations of the algorithm are in $\bigcup\limits_{s=1}^{i-1}N_G(u_s)$. Since $S$ is an open packing, $N_G(u_i)\cap N_G(u_s)=\emptyset$ for every $s< i$. These statements imply that no vertex in $N_G(u_i)$ is marked before the $i^{th}$ iteration. So, the algorithm will never encounter a marked vertex, and hence the algorithm will return `$S$ is an open packing in $G$'. Thus, the algorithm correctly determines the yes instances of the open packing sets of a graph.\\
	\emph{Case 2:} The input vertex subset $S$ is not an open packing in $G$.\\
	Since $S$ is not an open packing in $G$, there exist $i,s$ such that $1\leq i<s\leq k$ and $N_G(u_i)\cap N_G(u_s)\neq \emptyset$. W.L.O.G., assume that $u_s$ is the least indexed vertex such that there exists a vertex $u_i$ with $i<s$ and $N_G(u_i)\cap N_G(u_s)\neq \emptyset$. Note that $S'=\{u_1,\ldots,u_{s-1}\}$ is an open packing in $G$ by the choice of $s$. Hence, by the arguments in Case 1, the algorithm does not end before examining $N_G(u_s)$. Let $N_G(u_i)\cap N_G(u_s)=\{w_1,w_2,\ldots,w_p\}$ for some $p\leq n$. Note that every vertex in $N_G(u_i)$ is marked during $i^{th}$ iteration of the external while loop. So, the algorithm will encounter a marked vertex ($w_t$ for some $1\leq t\leq p$) while exploring $N_G(u_s)$ during the $s^{th}$ iteration. Hence, the algorithm will return `$S$ is not an open packing in $G$'. Thus, the algorithm correctly determines the no instances of the open packing sets of a graph.\\
	Thus, Cases 1 and 2 prove that Algorithm~\ref{testing-algo-for-op} correctly determines whether the input vertex subset $S$ is an open packing in $G$ or not.\\
	\noindent  Next, we prove that the algorithm runs in $O(n)$ time. To prove this, we show that the size of the list of vertices visited by the algorithm is at most $2n+1$. Let $Q=(v_1,v_2,\ldots,v_r)$ for some $r\in \mathbb{N}$. Then, it is clear that $Q$ is the list of vertices visited by the internal while loop (i.e., every vertex $v_j$ (for $1\leq j\leq r$) is visited as a neighbor of $u_i$ for some $i\in \{1,2,\ldots,k\}$) before the termination of the algorithm. We claim that $v_{j}\neq v_{\ell}$ for every $j,\ell\in \{1,2,\ldots,r-1\}$ and $j\neq \ell$. On the contrary, assume that there exist  $j,\ell\in \{1,2,\ldots,r-1\}$ such that $j\neq \ell$ but $v_{j}= v_{\ell}=u$ for some $u\in V(G)$. W.L.O.G., assume that $\ell$ is the least integer in $\{1,2,\ldots,r-1\}$ such that there exists $j<\ell$ with $v_{j}= v_{\ell}$. Then, the vertex $u(=v_{j})$ is marked by the algorithm when the internal while loop visits the $j^{th}$ vertex in $Q$. Also, by the assumption on $\ell$, every vertex in $(v_1,\ldots,v_{\ell-1})$ is distinct, and hence the algorithm does not terminate without visiting $\ell^{th}$ vertex of $Q$. But, when the $\ell^{th}$ vertex $u=v_\ell$ of the array $Q$ is visited by the internal while loop, the algorithm would terminate as the vertex is already marked while visiting the $j^{th}$ vertex ($u$) in $Q$, a contradiction to the assumption that  $r>\ell$ and the internal while loop visits $r$ vertices in $Q$ before the algorithm terminates. So, $v_{j}\neq v_{\ell}$ for every $j,\ell\in \{1,2,\ldots,r-1\}$ and $j\neq \ell$. This implies that the cardinality of the set $V'=\{v_1,v_2,\ldots,v_{r-1}\}$ is $r-1$. Since $V'\subseteq V(G)$, $r-1=|V'|\leq n$. Hence, $r\leq n+1$, i.e., the size of the list of vertices visited by the internal while loop is at most $n+1$. Note that for $Q=(v_1,v_2,\ldots,v_r)$ and $S=\{u_1,u_2,\ldots,u_k\}$, the list of vertices visited by the algorithm is a permutation of some subarray of $(v_1,v_2,\ldots,v_r,u_1,u_2,\ldots,u_k)$. We say subarray because the algorithm may terminate before visiting some vertices in $S$. Hence, the size of the list of vertices visited by the algorithm is at most $|S|+|Q|\leq 2n+1$. Also, note that the edges of $G$ are visited only during the internal while loop (for neighborhood testing). If $E'=(e_1,e_2,\ldots,e_q)$ is the list of edges visited by the algorithm with the suffices denoting the order in which the edges are visited, then by the construction of $Q$, it is clear that for every $e_i$ ($1\leq i \leq q$) visited by the algorithm, exactly one vertex ($v_i$) in $Q$ was visited. Since the size of $Q$ is at most $n+1$,  $q\leq n+1$. Hence, the algorithm runs in $3n+2=O(n)$ time.
\end{proof}
\section{Proofs related to Section~\ref{sec:sup-rk1}}\label{app-sec:claim_Gr}
 \begin{claim-alt}
	The graph $G_r$ defined in Remark~\ref{rem-tight-p4-rk1} is a $(P_4\cup rK_1)$-free graph.\\
	Recall that $V(G_r)=(\cup_{i=1}^r\{x_i,y_i,z_i\})\cup \{u,v\}$ and $E(G_r)=(\cup_{i=1}^r\{x_iy_i,y_iz_i,z_iu\})\cup \{uv\}$. On the contrary to claim, assume that there exists a $(P_4\cup rK_1)$, say $H$, in $G_r$. Let $\{a,b,c,d\}\in V(H)$ such that $G_r[{a,b,c,d}]\cong P_4$. Note that every $P_4$ in $G_r$ consists of the vertex $u$. So, $u\in \{a,b,c,d\}$. Similarly, observe that every $P_4$ in $G$ consists of exactly one vertex labelled $y_i$ for some $i=1,2,\ldots,r$. So, $y_i\in \{a,b,c,d\}$ for some $i=1,2,\ldots, r$. W.L.O.G., assume that $i=1$. Then, $\{a,b,c,d\}\in \{\{v,u,z_1,y_1\},\{z_j,u,z_1,y_1\},\{u,z_1,y_1,x_1\}\}$. Since $H\cong (P_4\cup rK_1)$ and $\{a,b,c,d\}\subseteq V(H)$ induces the $P_4$, $G_r[V(H)\setminus \{a,b,c,d\}]\cong rK_1$ and no vertex in $(V(H)\setminus \{a,b,c,d\})$ is adjacent to a vertex in $\{a,b,c,d\}$. So, for every possibility of $\{a,b,c,d\}$, $(V(H)\setminus \{a,b,c,d\})\subseteq \{x_2,y_2,\ldots, x_r,y_r\}$. But the independence number of $G_r[\{x_2,y_2,\ldots, x_r,y_r\}]$ is $(r-1)$, a contradiction to $(V(H)\setminus \{a,b,c,d\})\subseteq \{x_2,y_2,\ldots, x_r,y_r\}$ and $(V(H)\setminus \{a,b,c,d\})\cong rK_1$, i.e., the independence number of an induced subgraph of $G_r[\{x_2,y_2,\ldots, x_r,y_r\}]$ is at least $r$.
	\label{claim-gr-of-pr-rk1}
 \end{claim-alt}
 \begin{customlem}{\ref{p3rk1-op-bound}}
 		For $r\geq 1$, if $G$ is a connected ($P_3\cup rK_1$)-free graph, then $\rho^o(G)\leq 2r$.
 \end{customlem}
 \vspace{-3mm}
 \begin{proof}
 	We use induction on $r$ to prove this lemma as well. We begin with the proof that $\rho^o(G)\leq 2$ for every connected $(P_3\cup K_1)$-free graph $G$. Note that, by Lemma~\ref{p3rk1-td-bound}, $\rho^o(G)\leq \gamma_t(G)\leq 3$, for every connected $(P_3\cup K_1)$-free graph $G$. So, to prove the claim, it is enough to disprove the case $\rho^o(G)=3$. On the contrary, assume that there exists a $(P_3\cup K_1)$-free graph $G$ with $\rho^o(G)=3$. Since the open packing number of a connected $P_3$-free graph is at most two by Remark~\ref{rem-p3-free} and $\rho^o(G)=3$, $G$ must contain $P_3$ as an induced subgraph. Let $D=\{x_1,x_2,x_3\}$ be a subset of $V(G)$ that induces a $P_3$ in $G$ with $x_1x_2,x_2x_3\in E(G)$. Then, by the arguments in the proof of Lemma~\ref{p3rk1-td-bound}, $D$ is a total dominating set in $G$, i.e., $V(G)=\bigcup\limits_{i=1}^3N(x_i)$. Then, for every open packing $S$ in $G$, $S=\bigcup\limits_{i=1}^3(S\cap N(x_i))$. We know that $|S\cap N(x_i)|\leq 1$ for $i=1,2,3$. This, together with the fact that $S=\bigcup\limits_{i=1}^3(S\cap N(x_i))$, implies that for an open packing $S$ in $G$ of cardinality three, $|S\cap N(x_i)|=1$ for $i=1,2,3$. Let $ S\cap N(x_i)=\{z_i\}$. Then, $S=\{z_1,z_2,z_3\}$. So, $z_i\neq z_j$ for $i,j\in \{1,2,3\}$ and $i\neq j$ (else $|S|\leq 2$). Note that $S\cap N(x_1)=\{z_1\}$. So, $x_1z_2,x_1z_3\notin E(G)$. Using similar arguments, we can show that $x_iz_j\notin E(G)$ for $i,j\in \{1,2,3\}$ and $i\neq j$. Further, since $x_2x_3\in E(G)$ and $x_3z_1\notin E(G)$, $z_1\neq x_2 $. Similarly, $z_3\neq x_2$. Also, note that if $z_1z_3\in E(G)$, then $z_1z_2,z_2z_3\notin E(G)$ since $G[S]$ is $\{P_3,K_3\}$-free by Observation~\ref{obs-induced-s}. This, together with the fact that $z_3x_3\in E(G)$, implies that $z_2\neq x_3$. Then, $\{z_1,z_3,x_3\}\cup \{z_2\}$ induces a ($P_3\cup K_1$) in $G$, a contradiction. So, $z_1z_3\notin E(G)$. But, then $\{z_1,x_1,x_2\}\cup \{z_3\}$ induces a ($P_3\cup K_1$), a contradiction. So, $\rho^o(G)\leq 2$. Next, for every $r\geq 2$, we show that $\rho^o(G)\leq 2r$ for every connected $(P_3\cup rK_1)$-free graph $G$.\\
 	\textit{Induction Assumption}: Assume that for $r\geq 2$, $\rho^o(G')\leq 2(r-1)$ for every connected $(P_3\cup (r-1)K_1)$-free graph $G'$.\\
 	Let $G$ be a connected $(P_3\cup rK_1)$-free graph. If $G$ does not contain any vertex subset that induces a $(P_3\cup (r-1)K_1)$, then $G$ is $(P_3\cup (r-1)K_1)$-free, and so by induction assumption $\rho^o(G)\leq 2(r-1)< 2r$. So, assume that $G$ has a vertex subset $D=\{x_1,x_2,x_3,y_1,y_2,\ldots,y_{r-1}\}$ that induces a $(P_3\cup (r-1)K_1)$ in $G$ with $x_ix_{i+1}\in E(G)$ for $i=1,2$. Since $G$ is $(P_3\cup rK_1)$-free, every vertex in $V(G)\setminus D$ is adjacent to some vertex in $D$, i.e., $D$ is a dominating set in $G$. So, $V(G)=N[x_1]\cup N[x_2]\cup N[x_3]\cup N[y_1]\cup N[y_2]\cup \ldots \cup N[y_{r-1}]$. This implies that, for any open packing $S$ in $G$, $S=\bigcup\limits_{u\in D}(S\cap N[u])$. Note that $|S\cap N[u]|\leq |S\cap N(u)|+|S\cap \{u\}|\leq 2$. Also, since $x_1,x_2,x_3\in N(x_1)\cup N(x_2)$, $\bigcup\limits_{i=1}^3N[x_i]=\bigcup\limits_{i=1}^3N(x_i)$. So, $|S\cap (\bigcup\limits_{i=1}^3 N[x_i])|= |S\cap (\bigcup\limits_{i=1}^3 N(x_i))|\leq 3$. Therefore, $|S|=|S\cap V(G)|=|S\cap (N[x_1]\cup N[x_2]\cup N[x_3]\cup N[y_1]\cup N[y_2]\cup \ldots \cup N[y_{r-1}])|\leq 3+2(r-1)=2r+1$. Thus, to prove the claim, it is enough to contradict the possibility of $|S|=2r+1$. Assume that there exists an open packing $S$ of size $2r+1$ in $G$. Then, by the above arguments, we can conclude that $|S\cap N[y_j]|=2$ for $j=1,2\ldots,r-1$ and $|S\cap (\bigcup\limits_{i=1}^3 N(x_i))|= 3$. Since $|S\cap N(y_j)|\leq 1$ and $|S\cap N[y_j]|=2$, $y_j\in S$ and $S\cap N(y_j)\neq \emptyset$. Let $S\cap N[y_j]=\{w_j,y_j\}$ for $j=1,2,\ldots,r-1$ and let $S\cap (\bigcup\limits_{i=1}^3 N(x_i))=\{z_1,z_2,z_3\}$ with $x_iz_i\in E(G)$ (then, since $|S\cap N(x_i)|\leq 1$, we have $S\cap N(x_i)=\{z_i\}$ for every $i=1,2,3$). So, $S=\left(\bigcup\limits_{i=1}^3\{z_i\}\right)\cup \left(\bigcup\limits_{j=1}^{r-1}\{y_j\}\right)\cup \left(\bigcup\limits_{j=1}^{r-1}\{w_j\}\right)$. This, together with the fact that $|S|=2r+1$, implies that every vertex in the $(2r+1)$-tuple $(z_1,z_2,z_3,y_1,y_2,\ldots,y_{r-1},w_1,w_2,\ldots,w_{r-1})$ is distinct. Further, since (i) $S\cap N(x_i)=\{z_i\}$ and (ii) $S\cap N(y_i)=\{w_j\}$, we have $z_su\notin E(G)$ for $s\in \{1,2,3\}$ and $u\in D\setminus\{x_s\}$. We complete our claim in two cases based on $z_1z_3\in E(G)$ or not.\\
 	\emph{Case 1:} $z_1z_3\in E(G)$.\\
 	Then, by Observation~\ref{obs-induced-s}, $z_1z_2,z_2z_3\notin E(G)$. This implies that $z_2\neq x_3$. So, $\{x_1,z_1,z_3\}\cup \{z_2,y_1,y_2,\ldots,y_{r-1}\}$ induces a $(P_3\cup rK_1)$ in $G$, a contradiction.\\
 	\emph{Case 2:} $z_1z_3\notin E(G)$.\\
 	Then, $\{z_1,x_1,x_2\}\cup \{z_3,y_1,y_2,\ldots,y_{r-1}\}$ induces a $(P_3\cup rK_1)$ in $G$, a contradiction.\\
 	In both cases, we arrive at a contradiction. Therefore, $\rho^o(G)\leq 2r$.
 \end{proof}
 
\section{Proof Related to Section~\ref{sec:split}}\label{app:sec-split}
\begin{customclaim}{\ref{claim-s-ind-iff-sub}}
	A set $S(\subseteq V(G))$ is an independent set in the input graph $G$ of Construction~\ref{k-1,4-split-op-construct} if and only if $S$ is an open packing in the output graph $G'$.\\ Suppose that $S\subseteq V(G)$ is an independent set in $G$. On the contrary to Claim~\ref{claim-s-ind-iff-sub}, assume that $S$ is not an open packing in $G'$. Then, there exist vertices $u,u'\in S$ and $w\in V(G')$ such that $w\in N_{G'}(u)\cap N_{G'}(u')$. Since $u,u'\in V(G)\subseteq I$, $w\in C=E(G)\cup\{y,z\}$. By the construction of $G'$, $uy,uz\notin E(G')$ for any $u\in V(G)$. Thus, $w=e$ for some $e\in E(G)$. Then, by the construction of $G'$, $e=uu'\in E(G)$, a contradiction to $S$ being an independent set in $G$.\\
	Conversely, suppose that $S\subseteq V(G)$ is an open packing in $G'$. If $S$ is not an independent set in $G$, then there exist $u,u'\in S$ and an $e\in E(G)$ such that $uu'=e$. Then, by the construction of $G'$, $e\in N_{G'}(u)\cap N_{G'}(u')\neq \emptyset$, a contradiction.  
\end{customclaim}
\begin{theorem}
	\maxopenpack\ is hard to approximate within a factor of $N^{(\frac{1}{2}-\epsilon)}$ for any $\epsilon>0$ in $K_{1,4}$-free split graphs unless P = NP, where $N$ is the number of vertices in $K_{1,4}$-free split graphs.
	\label{k-1,4-inapprx}
\end{theorem}
\begin{proof}
	On the contrary, assume that \maxopenpack\ admits a $\displaystyle N^{\left(\frac{1}{2}-\epsilon\right)}$-factor approximation algorithm for some $\epsilon>0$ in $K_{1,4}$-free split graphs. Then, an open packing $S$ of the graph $G'$ constructed in Construction~\ref{k-1,4-split-op-construct} with $\displaystyle |S|> \frac{\rho^o(G)}{|V(G')|^{\left(\frac{1}{2}-\epsilon\right)}}$ can be found in polynomial time. Hence, an independent set $U$ of the input graph $G$ in Construction~\ref{k-1,4-split-op-construct} with $|U|=|S|-1$ can be found in polynomial time. Note that $N=|V(G')|=n+m+3\leq 2n+\frac{n^2}{2}\leq n^2$ for $n\geq 4$. Then,
	\begin{align*}
		|S|&>\frac{\rho^o(G')}{N^{\left(\frac{1}{2}-\epsilon\right)}}\hspace{2cm}\\
		|S|&>\frac{\rho^o(G')}{(n+m+3)^{\left(\frac{1}{2}-\epsilon\right)}}\hspace{2cm}\\
		2(|S|-1)&> \frac{\rho^o(G')}{(n^2)^{\left(\frac{1}{2}-\epsilon\right)}} \hspace{5mm} \text{for }|S|\geq 2\text{ and }n\geq 4\\
		|S|-1&>\frac{1}{2}\cdot\frac{\alpha(G)+1}{n^{1-2\epsilon}}\\
		|U|&>\frac{1}{2}\cdot\frac{\alpha(G)}{n^{1-2\epsilon}}
	\end{align*}
	Since $\rho^o(G')=\alpha(G)+1>\alpha(G)$ by Guarantee of Construction~\ref{k-1,4-split-op-construct}.\\
	Fix $n_0$ such that $\displaystyle n_0^\epsilon>2\implies\frac{1}{n_0^\epsilon}<\frac{1}{2}$.
	Then, for $n\geq \max\{n_0,4\}$, we have
	\begin{align*}
		|U|&>\frac{\alpha(G)}{n^\epsilon\cdot n^{1-2\epsilon}}\\
		|U|&>\frac{\alpha(G)}{n^{1-\epsilon}}
	\end{align*}
	The above inequality implies that \maxindset\ has a $(n^{1-\epsilon})$-factor approximation algorithm, which contradicts Theorem~\ref{hastard-ind-inapprx}.
\end{proof} 
\end{subappendices}
\end{document}